\documentclass{article}
\usepackage{graphicx} 
\usepackage{fullpage}
\usepackage{parskip}
\usepackage{natbib}

\usepackage{amsmath,eqnarray,mathtools, amsthm}
\usepackage{amsfonts, dsfont, bm, courier}
\usepackage{graphicx,xcolor,placeins,caption}
\usepackage{enumitem}
\usepackage{subcaption,afterpage}
\usepackage{natbib}
\usepackage{xcolor}
\usepackage[colorlinks,allcolors=teal]{hyperref}
\usepackage{cleveref}
\usepackage{algorithm}
\usepackage{algpseudocode}
\usepackage{booktabs}
\usepackage{longtable}
\usepackage{tikz-cd}
\usepackage{array}
\usepackage{comment}
\usepackage{thm-restate}
\usepackage{xr-hyper}

\usepackage{nicefrac}

\DeclareRobustCommand{\rauth}{%
  \raisebox{0.9ex}{%
    \tikz[baseline=(r.base)]{
      \node[
        draw,
        circle,
        inner sep=0.15pt,
        minimum size=0.7em,
        line width=0.35pt
      ] (r) {\fontsize{5}{5}\selectfont r};
    }%
  }%
}

\usepackage{tikz}
\usetikzlibrary{arrows.meta, shapes.geometric, calc,positioning,fit}
\usepackage[most]{tcolorbox}
\definecolor{boxblue}{RGB}{100, 160, 200}
\definecolor{boxbluedark}{RGB}{60, 120, 160}
\definecolor{boxred}{RGB}{210, 105, 105}
\definecolor{boxreddark}{RGB}{180, 75, 75}
\definecolor{boxgreen}{RGB}{130, 175, 120}
\definecolor{boxgreendark}{RGB}{90, 135, 80}
\definecolor{lightpurple}{HTML}{F4EEFF}
\definecolor{darkpurple}{HTML}{5B2A86}
\definecolor{softpurple}{HTML}{FAF7FF}
\definecolor{midpurple}{HTML}{8A5FBF}
\usepackage{placeins}

\definecolor{procIback}{HTML}{F2F8FC}
\definecolor{procIframe}{HTML}{A8CCE6}
\definecolor{procItitle}{HTML}{B9D9EF}
\definecolor{procItitletext}{HTML}{244A66}

\definecolor{procIIback}{HTML}{E0EDF7}
\definecolor{procIIframe}{HTML}{5E95C1}
\definecolor{procIItitle}{HTML}{4F88B7}

\definecolor{procIIIback}{HTML}{CCDDEA}
\definecolor{procIIIframe}{HTML}{285A7D}
\definecolor{procIIItitle}{HTML}{1F4E70}
\tcbset{
    procedurestyle/.style={
        enhanced,
        boxrule=0.7pt,
        arc=2.5mm,
        outer arc=2.5mm,
        left=8pt,
        right=8pt,
        top=6pt,
        bottom=5pt,
        width=0.92\linewidth,
        center,
        before skip=8pt,
        after skip=8pt,
        coltitle=white,
        fonttitle=\bfseries\sffamily,
        before upper={
            \setlength{\parskip}{2pt}
            \setlength{\parindent}{0pt}
        },
        attach boxed title to top left={
            yshift=-2.5mm,
            xshift=8mm
        },
        boxed title style={
            boxrule=0pt,
            arc=1.5mm,
            outer arc=1.5mm,
            left=6pt,
            right=6pt,
            top=2pt,
            bottom=2pt
        }
    }
}

\newtcolorbox{SimulatethenRoutebox}{
    procedurestyle,
    colback=procIback,
    colframe=procIframe,
    colbacktitle=procItitle,
    coltitle=procItitletext,
    title={I.\ Simulate--then--Route}
}

\newtcolorbox{RoutethenSimulatebox}{
    procedurestyle,
    colframe=procIIframe,
    colback=procIIback,
    colbacktitle=procIItitle,
    title={II.\ Route--then--Simulate}
}

\newtcolorbox{PredictRouteSimulatebox}{
    procedurestyle,
    colframe=procIIIframe,
    colback=procIIIback,
    colbacktitle=procIIItitle,
    title={III.\ Predictive Route--then--Simulate}
}

\newtcolorbox{promptbox}{
    enhanced,
    breakable,
    colback=softpurple,
    colframe=midpurple!80!black,
    boxrule=0.6pt,
    arc=2mm,
    outer arc=2mm,
    left=8pt,
    right=8pt,
    top=7pt,
    bottom=7pt,
    before skip=8pt,
    after skip=8pt,
    fontupper=\small
}

\DeclareMathOperator{\VCdim}{VCdim}

\theoremstyle{plain}
\newtheorem{theorem}{Theorem}
\newtheorem{proposition}{Proposition}
\newtheorem{lemma}{Lemma}
\newtheorem{corollary}{Corollary}

\theoremstyle{definition}
\newtheorem{definition}{Definition}
\newtheorem{assumption}{Assumption}
\newtheorem{example}{Example}

\theoremstyle{remark}

\newtheorem{claim}{Claim}

\begin{document}

\title{Social Choice Foundations for \\ Simulation-Augmented Generation}

\author{%
Sonja Kraiczy\textsuperscript{5,6}\rauth,
Smitha Milli\textsuperscript{1}\rauth,
Ratip Emin Berker\textsuperscript{2},
Avinandan Bose\textsuperscript{1,3}, \\
Brandon Amos\textsuperscript{1},
Jamelle Watson-Daniels\textsuperscript{1},
Maximilian Nickel\textsuperscript{1}, \\
Edith Elkind\textsuperscript{4},
Ariel D. Procaccia\textsuperscript{1,7}
\\[1ex]
\small
\textsuperscript{1}FAIR at Meta,
\textsuperscript{2}Carnegie Mellon University,
\textsuperscript{3}University of Washington, \\
\small
\textsuperscript{4}Northwestern University,
\textsuperscript{5}University of Oxford,
\textsuperscript{6}PrincInt,
\textsuperscript{7}Harvard University
\\[0.5ex]
\small \rauth \ Equal contribution; author order randomized.
}

\maketitle
\begin{abstract}
Simulation-augmented generation (SAGE) \citep{sage2026} is a recent technical proposal in which models simulate individuals' viewpoints at inference time in order to provide more representative answers to contentious user queries. A core challenge for SAGE is making inference-time simulation efficient without sacrificing representation quality. We introduce the first formalization of this problem, based upon an axiom from proportional clustering known as metric proportional justified representation+ (mPJR+) which is the strongest proportionality axiom known to always be satisfiable by centroid-based clustering. We prove that to proportionally represent the viewpoints of a population of $n_H$ humans on a given prompt, we need only create simulations of $n \ll n_H$ individuals, and at inference time, need only dynamically route to $k \ll n$ of those simulations based upon the prompt. This twofold reduction still yields approximate proportional representation guarantees for the entire population. Empirically, across two domains—political questions and personal advice—our proposed routing algorithm achieves higher mPJR+ satisfaction rates than $k$-means-based or random selection baselines.

\end{abstract}


\section{Introduction}

\citet{sage2026} recently propose \textbf{simulation-augmented generation} (SAGE) as a means of producing representative model responses to user queries on socially-contentious topics. At inference time, given a user prompt, the model queries generative simulations of individuals in a context-dependent target population and synthesizes their answers into a final response that is representative of the viewpoints of that population. For example, given the prompt, ``Should we have mandatory national military service?'', a model might identify an appropriate population, such as the user's country, and ask simulated individuals in that population what considerations matter most. Their responses might highlight concerns such as the need for defense preparedness or the opportunity cost of mandatory military service. The model then synthesizes these perspectives into a final response that reflects the range of viewpoints present in the target population.

Here we focus on a key technical challenge for achieving the vision of SAGE. Specifically, while much work in machine learning focuses on improving simulation fidelity~\citep{wu2026humanlm,naous2026flipping}, the shift to using simulations as an inference-time tool introduces an additional requirement: \emph{scalability}. A user-facing AI model must satisfy three constraints: \emph{latency}, responding to requests within acceptable timeframes; \emph{compactness}, managing context windows to support extended multi-turn conversations; and \emph{computational cost}, ensuring efficient resource utilization. To meet these constraints, \citet{sage2026} include a \emph{routing} step in which only a small number of individual simulations are actually queried at inference time (cf.~\Cref{fig:sage-architecture}), which in turn improves latency, reduces computational
costs, and yields a response set compact enough to fit in the context window.

The central challenge, though, is how to make the routing step efficient without sacrificing representativeness with respect to the full population. Concretely, we formulate the routing problem as follows. We are given a population of $n_H$ humans (e.g., millions, such as the population of a country) from which we construct simulations of a sample of $n \ll n_H$ individuals (e.g., $n = 1000$). Routing is the task of dynamically selecting, for a given prompt, $k \ll n$ of these individuals (e.g., $k = 10$) to simulate such that they satisfy formal representation guarantees with respect to the overall population of $n_H$ humans.\footnote{In this paper, we focus solely on the routing step, and leave the investigation of how these responses should subsequently be synthesized into a final answer for future work.}

\subsection{Overview of Approach}

\begin{figure}[t]
\includegraphics[width=\linewidth]{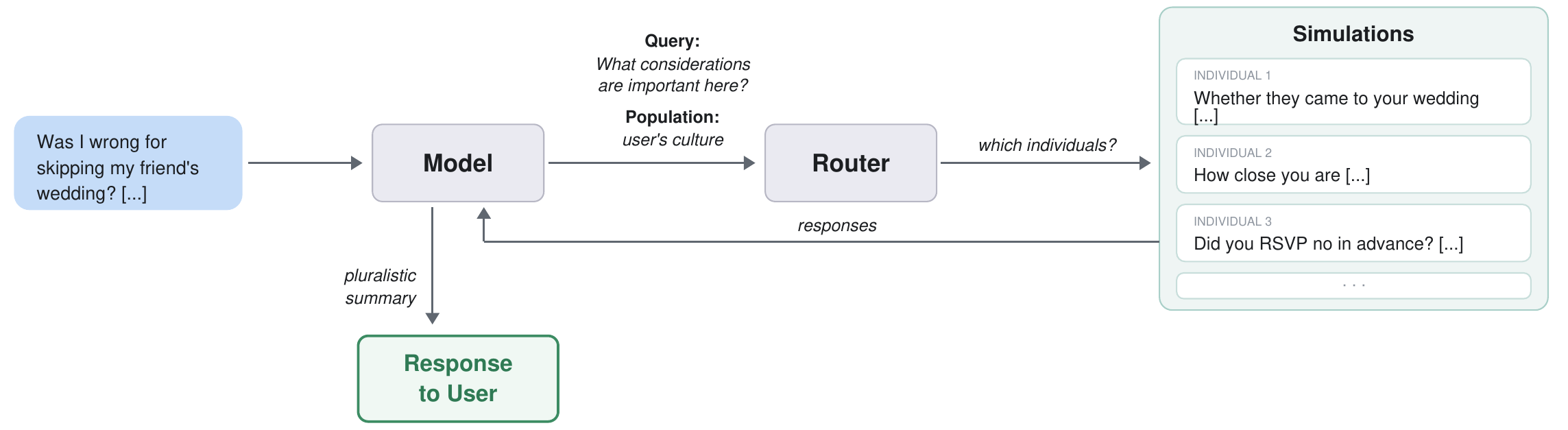}
\caption{\textbf{SAGE architecture.} Given a user query, the model identifies a prompt-appropriate target population and uses a \emph{router} to select which individual simulations to query. The selected simulations return responses that the model can synthesize into a final pluralistic answer.}\label{fig:sage-architecture}
\end{figure}
\textbf{Formalizing representation.} To formulate such guarantees, we turn to \emph{social choice theory}---a field which has long studied collective decision-making under diverse preferences. However, while classical social choice theory largely focuses on discrete and finite sets of alternatives, in our setting, simulations generate open-ended responses to a large and essentially unconstrained set of prompts. Preferences are therefore most naturally represented at the individual level via reward functions: each human $i$ has a reward function $r_i(p,o)$ that yields their reward from a response $o$ for prompt $p$.

The representation guarantee we then seek is: given a prompt $p$, we aim to select $k$ simulations to query so that their responses \emph{proportionally represent} the viewpoints of the overall population of humans. Informally, when a large group of humans has similar reward for some responses to this prompt, then a proportionate fraction of the $k$ simulated responses should deliver appropriately high reward to this group. Moreover, the reward level that is guaranteed to the group depends on how cohesive their preferences are: groups with more aligned preferences can be guaranteed higher reward.

\textbf{Efficiently guaranteeing representation.} We prove that to represent a population of $n_H$ humans, we need only simulate a random sample of $n \ll n_H$ individuals where $n$ grows logarithmically in $n_H$. We further prove that at inference time, querying just $k \ll n$ prompt-dependent simulations suffices to achieve approximate representation guarantees. To accomplish this, we adapt proportional clustering algorithms to our setting and use them to route to $k$ simulations given a prompt. The degree of representation achieved by their responses depends on the fidelity gap between all simulations and the humans they represent. Empirically, when we have black-box simulation access and do not have access to rewards, the degree of representation additionally depends on the error that arises from predicting the simulations' viewpoint embeddings.

\subsection{Our Contributions}

Below we state our technical contributions.

\begin{enumerate}[leftmargin=15pt]
    \item \textbf{Formal representation guarantees.} We provide the first formalization of representative routing in SAGE through the careful adaptation of axioms from the computational social choice literature, chiefly \emph{metric proportional justified representation plus (mPJR+)}~\citep{Kellerhals24:Proportional,he2026checkpleaseverifiablyfair}.

    \item \textbf{Efficient proportional clustering.} We develop a faster implementation of the spatial expanding approval rule (SEAR),\footnote{More precisely, we consider a variant of this algorithm, discussed in \Cref{sec:geom-rep} and \Cref{app:geom-rep}.} a proportional clustering algorithm which satisfies the representation axiom we study. Our implementation achieves an $O(n^2 \log n)$ amortized running time, improving upon the previous $O(k^2 n^4)$\footnote{Versions of this paper have been presented at several workshops, including EEAMO '25, IASEAI '26, and ICML '26. The cited result refers to the version of the paper available at the time we established our result. In July 2026, the authors updated their paper with a runtime of \(O(n^2(\log n + k))\).} bound, which makes it practical for inference-time deployment.
    
    \item \textbf{Efficient querying of simulations.} Under the assumption of a factorized reward structure across individuals, we prove that an adaptation of proportional clustering yields an efficient inference-time procedure that only requires querying $k \ll n$ simulations for a given prompt. The responses from these $k$ simulations are approximately representative of the full population of simulations, with the level of approximation depending on the fidelity gap between simulations and their corresponding humans.

    \item \textbf{Scaling from simulations to populations.} Our ultimate target is proportional representation of an underlying human population of size $n_H$, which may be far larger than the simulated population size $n$. We prove that it suffices to run our method on a population of $n \ll n_H$ simulations that simulates a random (possibly stratified) sample of the true population: specifically, $n$ need grow only logarithmically in $n_H$, and when the reward model dimension $d$ is small, $n$ can depend only on $d$ and be independent of $n_H$ and $k$.
    
    \item \textbf{Empirical validation.} We evaluate representative routing in two domains: simulating political viewpoints on Remesh~\citep{konya2023democratic} and personal advice from Reddit's ``Am I the Asshole?'' subreddit. Since we only have black-box simulation access, and no access to human rewards, our experiments evaluate representativeness with respect to the pool of $n$ simulations. Across both domains, our proposed algorithm achieves higher mPJR+ satisfaction than $k$-means or uniform random selection baselines.
    
\end{enumerate}

\subsection{Related Work} \label{sec:related_work}

\textbf{Social choice.} Beginning with the work of \citet{aziz2017justified}, an influential line of research in social choice has formalized notions of proportional representation in committee elections and participatory budgeting~\citep{AEHL+18,PPS21}. Broadly speaking, these notions ensure that every sufficiently large subset of voters with sufficiently cohesive preferences is represented in the outcome. 
A similar approach to proportional representation has also been explored in the context of clustering~\citep{Chen20:Proportionally,prf,Kellerhals24:Proportional}, where the same underlying idea is given a geometric formulation: agents and representatives lie in a metric space, and proportional representation is defined in terms of the distances between them. We adopt this geometric perspective as well. Our treatment is closely related to the proportional representative fairness (PRF) property, introduced by \citet{prf}, and in particular builds on mPJR+ ~\citep{Kellerhals24:Proportional,he2026checkpleaseverifiablyfair}.

Our work is similar in spirit to that of \citet{gensocialchoice} and \citet{boehmer2025generative}. Their goal, too, is to generate a slate of viewpoints that is formally representative of the views of a given set of simulations. Their approach uses \emph{discriminative queries} and \emph{generative queries} as building blocks. Realizing the latter, however, requires numerous calls to an LLM. Thus, in the context of SAGE, their approach provides an alternative way to satisfy the representation criterion, but does not satisfy the inference-time scalability criterion.

In general, most works on social choice for AI alignment focus on aggregation of preferences or rewards prior to deployment~\citep{ge2024,
conitzer2024,golz2025distortion,hokr2026,maurarivero2025jackpotalignmentmaximallottery}, rather than at inference time as our approach does. A recent exception is  concurrent work by \citet{freedman2026adaptivepluralisticalignmentpipeline}, who proposes generating diverse candidate responses at inference time and then applying a social choice voting procedure with heterogeneous reward functions to select the winner.

\textbf{Pluralism in LLMs.} Simulation-augmented generation~\citep{sage2026} is a means of creating model responses that are representative of a range of perspectives, which aligns with several threads in alignment, such as Overton pluralism in pluralistic alignment~\citep{sorensen2024}, reasonable pluralism~\citep{fisher2025position}, and work on mitigating disempowerment~\citep{sharma2026s}. While LLM-based simulations have been developed for other purposes—such as replicating social science experiments~\citep{hewitt2024predicting,aher2023,kolluri-etal-2025-finetuning} or serving as user simulators for model evaluation~\citep{yao2025taubench}—SAGE introduces unique scalability requirements by deploying simulations as inference-time calls in user-facing systems. A closely related precursor is \emph{modular pluralism} by \citet{feng-etal-2024-modular}, which queries $k$ ``community LLMs'' at inference time to surface diverse perspectives. However, their approach relies on a fixed, pre-specified set of community models, which cannot adapt to the full diversity of user queries. Our work addresses this limitation: we maintain the scalability of querying only $k$ simulations while enabling adaptive, prompt-specific selection from a larger population of $n$ simulated individuals, with formal guarantees that the selected simulations proportionally represent the broader population.

\section{Inference-Time Social Choice with Simulations}
\label{sec:reward-commutation}
In this section, we study representative routing when human reward functions
are available at inference time. For each human, we assume access to both a
reward function over responses and a simulation that generates responses on
that human's behalf. We first formalize what it means for a slate of simulated
responses to represent a population at inference time. We then introduce
\emph{Simulate--then--Route}, a conceptually complete approach that queries every
simulation and subsequently selects a representative slate, drawing on notions of
proportional representation from computational social choice \citep{aziz2017justified}.
This exhaustive approach exposes three scalability bottlenecks. First,
representative selection itself must be computationally efficient, so in \Cref{sec:simulate-then-route} we give an
implementation with running time nearly linear in the input
size. Second, maintaining a simulation for every member of a large target
population is unrealistic. In \Cref{sec:geom-rep} we show that, under shared structure in the reward
functions, representation with respect to a much smaller (random) simulation
pool transfers to the underlying population with high probability. Finally,
querying every simulation at inference time can be prohibitively expensive and
slow. We show that routing and simulating approximately commute in \Cref{sec:route-then-simulate}, allowing us
to first select $k\ll n$ representative simulations and query only those,
while preserving proportional representation approximately.

\subsection{Representation under Reward Access}
\label{sec:simulate-then-route}

Fix a prompt $p$ and a response space $\mathcal O$. Let $\mathcal{N}$ be a set of
$n$ humans for which we have access to their reward function
$r_i(p,\cdot):\mathcal O\to\mathbb R$ and a simulation $A_i$, which given prompt $p$ simulates human $i$'s response $o_i:=A_i(p)$. We use these to define a candidate set of responses, namely the indexed collection
$C:=(o_i)_{i\in \mathcal{N}}$ so that even if simulated responses coincide they remain distinct candidates. Our goal is to select a slate of $k$ from these candidates.

For $\ell\in[k]$, we say that a group $S\subseteq \mathcal{N}$ is \emph{$\ell$-large} if $|S|\ge \ell n/k$; since such a group is an $\ell/k$ fraction of the entire population, it has a proportional claim to $\ell$ of the $k$ slate positions. The strength of this claim depends on how strongly the group can agree on $\ell$ responses. In particular, if for an unselected response $c$ every member of $S$ assigns $c$ reward at least $t$, then reward PJR+ (formalized below) requires the slate to contain at least $\ell$ responses that give reward at least $t$ to some member of $S$. Thus, larger groups are entitled to proportionally more responses, while groups that share a highly valued common response are guaranteed representation at a correspondingly high reward level.

\begin{definition}
A slate $Y\subseteq C$ of $k$ responses satisfies
\emph{reward proportionally justified representation}  (rPJR+) on
prompt $p$ if, for every $\ell\in[k]$, every unselected response
$c\in C\setminus Y$, and every $\ell$-large group $S\subseteq \mathcal N$, at least
$\ell$ responses $o\in Y$ satisfy $\max_{i\in S}r_i(p,o)\ge
\min_{i\in S}r_i(p,c)$.
\label{def:it-prf}
\end{definition}
Equivalently, reward PJR+ rules out a simple form of justified deviation: an
$\ell$-large group cannot point to an unselected response $c$ that they all
value at least $t$ while fewer than $\ell$ selected responses attain value
$t$ for any member of the group.
Moreover, adding further responses to $C$ only strengthens reward PJR+ as it gives the groups more flexibility to deviate.

\emph{Simulate--then--Route} provides a direct approach to achieve this representation objective: first query every simulation to obtain a response, then select a representative slate of $k$ responses from the resulting candidate set. Here, routing determines which generated responses to retain. Access to the realized responses and the human reward functions allows this selection to be based directly on the rewards $r_i(p,o_j)$ that each human assigns to each candidate. Specifically, \textsc{ResponseRouter} takes the reward functions $(r_i)_{i\in\mathcal N}$, the generated responses $(o_i)_{i\in\mathcal N}$, and the target slate size $k$, and selects the slate using the Reward Expanding Approvals Rule (REAR), described below.

\begin{SimulatethenRoutebox}
\label{proc:simulate-then-route}
\raggedright

\textbf{Input:} prompt \(p\), reward functions
\(\{r_i\}_{i\in\mathcal{N}}\),
simulations \(\{A_i\}_{i\in\mathcal{N}}\), slate size \(k\).

\textbf{Output:} slate \(Y\) of \(k\) simulated viewpoints.

\medskip

\begin{enumerate}[leftmargin=*, itemsep=0.5em]

    \item \textbf{Simulate all responses.}
    Query every simulation \(i\in\mathcal{N}\) to obtain response
    \(
        o_i := A_i(p).
    \)

    \item \textbf{Select $k$ responses to route to.} Return \(Y:= \textsc{ResponseRouter}((r_i)_{i\in \mathcal{N}},(o_i)_{i\in \mathcal{N}},k).\)
\end{enumerate}
\end{SimulatethenRoutebox}

\paragraph{Reward Expanding Approvals Rule (REAR).}
Given reward functions $(r_i)_{i\in\mathcal N}$, an indexed collection of
 responses $O=(o_1,\ldots,o_m)$, and a target slate size $k$, each
human $i$ starts with budget $b_i=1$, and the slate is initialized as empty.

We continuously decrease the reward threshold; at reward threshold $\tau$, human $i$ approves candidate $o_j$ whenever
$r_i(p,o_j)\ge\tau$. Whenever the approvers of an unselected candidate have
total remaining budget at least $n/k$, that candidate is added to the slate and
$n/k$ units of budget are deducted in total from its approvers. Ties are
resolved according to a fixed rule. The procedure terminates after selecting
$k$ responses.

\begin{restatable}[Representation Guarantee]{theorem}{gtsitprf}
\label{thm:gts-it-prf}
Simulate--then--Route returns
a slate satisfying reward PJR+.
\end{restatable}

From an algorithmic perspective, the main issue is efficiency, since we
ultimately want to enforce representation guarantees at inference time while
respecting latency constraints.
\citet{aziz2020expanding} show that the Expanding Approvals Rule (EAR), which
corresponds to the special case in which rewards encode ordinal rankings on a
common scale,\footnote{Their tie-breaking rule differs from ours, but this
choice does not affect the axiomatic guarantees that we consider.}
can be implemented in $O(nm^2)$ time.
Moreover, \citet{prf} analyze a closely related algorithm, the Spatial
Expanding Approvals Rule (SEAR), in which voter--candidate scores are induced
by distances and are therefore more general than ordinal rankings. Their
implementation runs in $O(k^2n^4)$ time in the case where $m=n$.

For our application, these running times would be prohibitive even for
simulated populations of moderate size, such as $n\approx 10^3$. We show that
REAR admits a substantially faster and more streamlined implementation,
which is crucial in our inference-time setting and makes this algorithmic
approach practical in our framework.
The improvement comes from more efficient bookkeeping, together with choices of
tie-breaking and payment rules that permit a better amortized running time.

\begin{restatable}{theorem}{searruntime}
\label{thm:rear-runtime}
REAR can be implemented in
\(O(mn\log n)\) time.
\end{restatable}

Two limitations remain. First, our guarantee currently applies only to the humans for whom we have simulation access, rather than the larger target population, motivating the population-transfer results in \Cref{sec:geom-rep}. Second, Simulate--then--Route requires $n$ expensive simulation queries to return only $k$ responses; when simulations are implemented by large autoregressive models, potentially with inference-time reasoning, this can be prohibitive in both latency and compute, motivating the commutation results in \Cref{sec:route-then-simulate}.
\subsection{Geometric Representation under Factorized Rewards}
\label{sec:geom-rep}

Simulate--then--Route performs routing directly according to the humans'
rewards over the realized responses. 
We now consider reward functions
that factor as an inner product between a shared response embedding
and a prompt-conditioned individual preference embedding. This allows
heterogeneous preferences to be expressed relative to a common response
space. The use of such shared representations
with individual preference parameters is motivated by personalized reward
modeling \citep{chen2025pal,bose2025lorepersonalizingllmslowrank} and by matrix-factorization
approaches to collaborative filtering \citep{koren2009matrix}.
We formalize this structure, together with normalization and realizability
conditions, as follows.

\begin{assumption}[Factorized rewards with a shared response basis]
\label{assump:factorized-reward}
There are maps
\(\psi:\mathcal O\to\mathbb S^{d-1}\) and
\(\phi_i:\mathcal P\to\mathbb S^{d-1}\), \(i\in\mathcal N\), such that, for
every prompt \(p\) and response \(o\),
\[
    r_i(p,o)=\phi_i(p)^\top\psi(o).
\]
Moreover, for every $i\in\mathcal N$ and prompt $p$, there exists a response
$o_i^*\in\mathcal O$ such that $\psi(o_i^*)=\phi_i(p)$.
\end{assumption}

This assumption imposes a shared representation of responses through $\psi$
while still allowing each human to have their own prompt-conditioned reward
embedding $\phi_i(p)$. 
Throughout this section, the maps in Assumption~\ref{assump:factorized-reward} are available for evaluation.

\paragraph{From Sampled Simulations to Populations.}
Simulate--then--Route provides representation guarantees for the $n$
humans in $\mathcal N$ for whom we obtain responses. We now consider a
larger target population $\mathcal H$, of size $n_H:=|\mathcal H|$.
Obtaining a response for every member of a population as large as the
United States may be impractical; instead, we ask whether it suffices to
obtain responses for only a much smaller random sample.

Here we assume that Assumption~1 holds more generally for every human in $\mathcal H$, with
the same response map $\psi$. For each human $i\in\mathcal H$ and prompt
$p$, fix a response $o_i^*(p)\in \mathcal{O}$ satisfying
$\psi(o_i^*(p))=\phi_i(p)$. For a fixed prompt $p$, let
$C_{\mathcal H}:=(o_i^*(p))_{i\in\mathcal H}$ denote the corresponding
indexed full-population candidate collection.

We draw $\mathcal N$ uniformly from the $n$-element subsets of
$\mathcal H$ and obtain only the responses of the sampled humans, giving
$C_{\mathcal N}:=(o_i^*(p))_{i\in\mathcal N}$. Thus, the slate is selected
from $C_{\mathcal N}$, while $C_{\mathcal H}$ serves as the conceptual
population-level benchmark. Importantly, the responses
$o_i^*(p)$ for $i\notin\mathcal N$ need not be obtained in order to
construct the slate.

Exact reward PJR+ need not transfer from the sample to the full
population. We therefore allow slack both in the size of groups that
receive a proportional claim and in the reward level at which that claim
is satisfied. For $\varepsilon\in[0,1]$, we say that a slate
$Y\subseteq C_{\mathcal N}$ satisfies
\emph{$(\varepsilon,\gamma_\beta)$-approximate reward PJR+} with respect
to $(\mathcal H,C_{\mathcal H})$ on prompt $p$ if, for every
$\ell\in[k]$, every comparison response
$c\in C_{\mathcal H}\setminus Y$, every $\beta\in[0,\pi]$, and every
group $S\subseteq\mathcal H$ satisfying
$|S|\geq(\ell/k+\varepsilon)n_H$ and
$r_i(p,c)\geq\cos(\beta)$ for all $i\in S$, at least $\ell$ responses
$o\in Y$ satisfy
$\max_{i\in S}r_i(p,o)\geq\cos(\gamma_\beta)$.
When $\varepsilon=0$, we simply write
$\gamma_\beta$-approximate reward PJR+. When additionally
$\gamma_\beta=\beta$, this recovers exact reward PJR+ with respect to
$(\mathcal H,C_{\mathcal H})$.

The following theorem shows that obtaining responses from only a random
sample of the population is sufficient to guarantee approximate
representation with respect to the full population and its full candidate
collection.

\begin{restatable}[Population Transfer]{theorem}{sampledpop}
\label{thm:sampled-population}
Fix a prompt $p$ and let $\varepsilon,\delta\in(0,1)$. For each
$i\in\mathcal H$, fix a response $o_i^*(p)$ satisfying
$\psi(o_i^*(p))=\phi_i(p)$. Let
$n=\left\lceil
\frac{1}{2\varepsilon^2}\log\frac{2kn_H}{\delta}
\right\rceil$ and draw $\mathcal N$ uniformly from the $n$-element subsets of $\mathcal H$
and let $C_{\mathcal N}:=(o_i^*(p))_{i\in\mathcal N}$. Then, with probability at least $1-\delta$, every size-$k$ slate
$Y\subseteq C_{\mathcal N}$ satisfying reward PJR+ with respect to the
sampled instance $(\mathcal N,C_{\mathcal N})$ satisfies
$(\varepsilon,\gamma_\beta)$-approximate reward PJR+ with respect to
$(\mathcal H,C_{\mathcal H})$, where
$\gamma_\beta=\min\{\pi,4\beta\}$.
\end{restatable}

For example, taking $n_H=8\text{ billion}$, $k=10$, $\delta=0.05$, and $\varepsilon=0.03$ gives a sample size of roughly $9\times10^3$. Thus, the number of simulations needed to represent a large population can grow only logarithmically with the population size. In \Cref{app:sampling}, we also give a dimension-dependent bound for arbitrary groups where the required sample size is independent of both $n_H$ and $k$.

\paragraph{Proportional Clustering}
Assumption \ref{assump:factorized-reward} allows us to relate the reward-based routing problem from \Cref{sec:simulate-then-route} to the proportional clustering literature in social choice \citep{Chen20:Proportionally, prf, he2026checkpleaseverifiablyfair}. This insight will ultimately allow us to query only $k\ll n$ simulations while preserving representation in \Cref{sec:route-then-simulate}.

For a fixed prompt \(p\), let
\( v_i:=\phi_i(p)\) and \(
    c_j:=\psi(o_j).\)
Thus \(v_i\) is human \(i\)'s prompt-conditioned preference embedding, while
\(c_j\) is the embedding of candidate response \(o_j\). Let
\(V=(v_1,\ldots,v_n)\), let \(C=(c_1,\ldots,c_m)\), and let \(\Delta\)
denote Euclidean distance. We treat voters and candidates as indexed objects,
so distinct objects may occupy the same point.
The geometric analogue of reward PJR+ is the following proportional
representation criterion \citep{he2026checkpleaseverifiablyfair}.

\begin{definition}[mPJR+]
\label{def:prf}
A slate \(X\subseteq C\) of size \(k\) satisfies \emph{metric proportional
justified representation plus} (mPJR+) if, for every \(\ell\in[k]\), every
unselected candidate \(c\in C\setminus X\), and every \(\ell\)-large group
\(S\subseteq\mathcal N\), at least \(\ell\) members \(x\in X\) satisfy
\[
    \min_{i\in S}\Delta(v_i,x)
    \le
    \max_{i\in S}\Delta(v_i,c).
\]
\end{definition}

If we want to specify the candidate set and voter set we say $X$ satisfies mPJR+ w.r.t.\ $(V,C)$, or w.r.t.\ $V$ if $V=C$.

In the metric setting, to route $k$ points satisfying our representation guarantees, we use a
\emph{proportional} clustering algorithm--- a variant of \emph{spatial expanding
approval rule} (SEAR)\footnote{The proportional clustering literature contains
several closely related variants of this procedure of greedily picking groups to
be represented and deactivating them, sometimes referred to as \emph{greedy
capture}. Although these variants share the same basic structure, the precise
implementation matters for the axiomatic guarantees obtained. We discuss these
algorithmic differences in the appendix.} \citep{prf}.
Unlike (balanced) $k$-means, SEAR (\Cref{alg:sear}) provably returns $k$ centroids satisfying
mPJR+ \citep{he2026checkpleaseverifiablyfair}. For examples showing that $k$-means and balanced $k$-means fail mPJR+, see \Cref{app:geom-rep}.
\begin{algorithm}[h] \caption{Spatial Expanding Approval Rule (SEAR)} \label{alg:sear} \begin{algorithmic}[1] \Require Voters $V=(v_1,\ldots,v_n)$, candidates $C=(c_1,\ldots,c_m)$, slate size $k$ \State Initialize budgets $b_i\gets 1$ for all $i\in[n]$ and $X\gets\emptyset$ \State Grow balls of radius $\rho$ around all candidates at a common rate \While{$|X|<k$} \State Increase $\rho$ until some $c_j\notin X$ satisfies $\sum_{i:\Delta(v_i,c_j)\le\rho} b_i\ge n/k$ \State Add $c_j$ to $X$ and reduce the budgets of its approving voters by $n/k$ in total \EndWhile \State \Return $X$ \end{algorithmic} \end{algorithm}

Under the factorized reward structure, reward PJR+ and mPJR+ are intimately
connected: reward thresholds translate exactly into distance thresholds in the
shared embedding space. The same correspondence also holds at the algorithmic
level, relating REAR to SEAR. Formally:
\begin{restatable}[Reward--metric equivalence]{proposition}{rewardmetricequivalence}
\label{prop:reward-metric-equivalence}
Fix a prompt \(p\) and a collection of candidate responses
\(O=(o_1,\ldots,o_m)\). Under
Assumption~\ref{assump:factorized-reward}, let \(V=(\phi_i(p))_{i\in\mathcal N}\) and \(C=(\psi(o_j))_{j\in[m]}\).
Then:
\begin{enumerate}
    \item a response slate \(Y\subseteq O\) satisfies reward PJR+ if and only if
    the corresponding indexed embedding slate satisfies mPJR+ with respect to
    \((V,C)\); and
    \item under corresponding tie-breaking and payment rules, REAR on
    \((O,(r_i)_{i\in\mathcal N},k)\) and SEAR on \((V,C,k)\) select the same
    indexed candidates.
\end{enumerate}
\end{restatable}
 Algorithmically, \Cref{thm:rear-runtime} then immediately yields a
substantially faster implementation of SEAR than the previous
\(O(k^2n^4)\)-time implementation of \citet{prf}.

\begin{corollary}
SEAR can be implemented in time $O(n^2 \log n)$.
\end{corollary}
Thus, under Assumption~\ref{assump:factorized-reward},
Simulate--then--Route can equivalently be viewed as first generating all
response embeddings \(C=(\psi(A_i(p)))_{i\in\mathcal N}\) and then running
SEAR with \(V=(\phi_i(p))_{i\in\mathcal N}\) as voters and \(C\) as
candidates. 
\subsection{Simulating and Routing Commute}
\label{sec:route-then-simulate}

Under factorized reward access, we can reverse the order of simulating responses and selecting responses: because the prompt-conditioned preference embeddings $\phi_i(p)$ can
be evaluated without querying the simulations, we can first route to $k$
representative humans and only then simulate their responses. We show that
these two orders commute exactly when the simulations have zero fidelity gap,
and approximately when the fidelity gap is nonzero.

\paragraph{\textsc{SimulationRouter}.}
Given an indexed collection of representations
$Z=(z_i)_{i\in\mathcal N}$, one for each simulation, and a target slate
size $k$, \textsc{SimulationRouter} runs SEAR using $Z$ as both the voter
and candidate multiset and returns the indices of the selected candidates:
\[
    \textsc{SimulationRouter}(Z,k)
    := \{i\in\mathcal N : z_i \text{ is selected by }
        \operatorname{SEAR}(Z,Z,k)\}.
\]
\begin{RoutethenSimulatebox}
\label{proc:route-then-simulate}
\raggedright

\textbf{Input:} prompt \(p\), reward functions
\(\{r_i\}_{i\in\mathcal{N}}\), where
\(
    r_i(p,o)=\phi_i(p)^\top\psi(o),
\)
simulations \(\{A_i\}_{i\in\mathcal{N}}\), slate size \(k\).

\textbf{Output:} slate \(Y\) of \(k\) simulated viewpoints.

\medskip

\begin{enumerate}[leftmargin=*, itemsep=0.5em]

    \item \textbf{Read off preference embeddings.}
    Set
    \(
        v_i:=\phi_i(p)
    \)
    for every \(i\in\mathcal{N}\).

    \item \textbf{Select \(k\) simulations to route to.}
    $R:=\textsc{SimulationRouter}((v_i)_{i\in \mathcal{N}},k)$

    \item \textbf{Simulate selected simulations.}
    Query only the simulations in \(R\) and return
    \[
        Y := \{A_r(p): r\in R\}.
    \]

\end{enumerate}
\end{RoutethenSimulatebox}

We first consider the ideal case in which each simulation returns a reward-maximizing response for its corresponding human. As the next theorem shows, in this ideal case, the costly generations for unselected humans can be skipped without losing proportional representation: any large, cohesive group that would be represented by simulating all responses is already represented after selecting in the prompt-conditioned embedding space, routing to the corresponding simulations, and only simulating them.

\begin{restatable}[Exact Commutation]{theorem}{exactstg}
\label{thm:exact-stg}
Fix a prompt \(p\) and suppose that, for every \(i\in\mathcal N\),
\[
A_i(p)\in\arg\max_{o\in\mathcal O} r_i(p,o).
\]
Then, under a common tie-breaking rule, Simulate--then--Route and
Route--then--Simulate select the same set of simulation indices.
Consequently, Route--then--Simulate satisfies reward PJR+ while querying only
\(k\) simulations.
\end{restatable}

For imperfect simulations, let
$\operatorname{Simgap}(p):=\max_{i\in\mathcal N}
\angle\!\left(\phi_i(p),\psi(A_i(p))\right)$ 
denote the simulation fidelity gap on prompt $p$. Under
Assumption~\ref{assump:factorized-reward},
$\operatorname{Simgap}(p)\le\alpha$ equivalently means
$r_i(p,A_i(p))\ge\cos(\alpha)$ for every
$i\in\mathcal N$. Faithful simulation corresponds to $\alpha=0$. In general, when the fidelity gap is nonzero, the two orders need not coincide exactly and route--then--simulate may fail exact reward PJR+. Nevertheless, they
approximately commute in terms of representation quality. In particular, if
$\operatorname{Simgap}(p)\leq\alpha$, routing before simulation incurs only
the following approximation of the representation guarantee.

\begin{restatable}[Approximate Commutation]{theorem}{approxstg} \label{thm:approx-stg} Fix a prompt $p$ and suppose $\operatorname{Simgap}(p)\le\alpha$. The slate returned by Route--then--Simulate satisfies $\gamma_\beta$-approximate reward PJR+ for $\gamma_\beta=\alpha+\beta+\min\{\alpha,\beta\}$. \end{restatable}

\section{Representation under Viewpoint Embedding Access}
\label{sec:predicted-response-embeddings}

So far, we have assumed access to the reward maps $r_i$ at inference time. We now consider a weaker setting in which these individual reward functions are unavailable and the only individualized access to human $i$ is through black-box queries to their simulation $A_i$. This reflects an emerging deployment model for commercial simulations, in which simulated individuals or populations are accessed through provider-operated services without access to the personalized models or representations underlying their responses \citep{simile2026simulation}. 

We retain access to the embedding function $\psi$ from Assumption~\ref{assump:factorized-reward}, which we refer to as the \emph{viewpoint embedding (VPE) function}. The defining property of the VPE is that it embeds responses such that those expressing similar human beliefs or preferences are nearby. Access to such an embedding function is a substantially more viable assumption in practice: recent work provides representations of human viewpoint similarity learned from patterns of human belief agreement \citep{lee2025semantic} and, more directly, representations trained so that embedding distance reflects \emph{preferential} rather than merely semantic similarity \citep{prefembeddings2026}; related work similarly learns stance-aware representations that separate agreeing from opposing viewpoints \citep{ghafouri2024love}.
We thus treat the VPE as an off-the-shelf component supplied by existing viewpoint-embedding approaches both in this section and in our experiments.

\paragraph{Exhaustive Viewpoint Embedding Geometry.}Under VPE access, one way to obtain a proxy for the personalized geometry is to query every simulation and compute \[
    c_i(p) := \bigl(\psi \circ A_i\bigr)(p)
\qquad \text{for all  }i\in\mathcal N,\]
and then run \textsc{SimulationRouter} on these embeddings. The zero-fidelity-gap case gives this geometry a particularly useful
interpretation. When $\alpha=0$,
$c_i(p)=\psi(A_i(p))=\phi_i(p)$ for every $i\in\mathcal N$. Thus, in this
case, generating the simulation's response and embedding it with $\psi$
recovers exactly the prompt-conditioned preference embedding $\phi_i(p)$ used under reward
access. 

More generally, routing on the exhaustive VPE geometry yields the following
guarantee even when the fidelity gap is nonzero. 
\begin{proposition}
\label{prop:psipsi}
Fix a prompt $p$ and suppose $\operatorname{Simgap}(p)\leq \alpha$. The slate $Y=\{A_r(p): r\in R\}$ where $R=\textsc{SimulationRouter}((c_i)_{i\in \mathcal{N}},k)$ satisfies 
$\gamma_\beta$-approximate reward PJR+ for
$\gamma_\beta=\beta+2\alpha$.
\end{proposition}
In particular, when $\alpha=0$, this recovers exact reward PJR+.

\paragraph{Predicting the Viewpoint Geometry.}
The exhaustive procedure above requires querying every simulation, precisely the
complete generation step that we seek to avoid. Instead, we learn a cheaper prompt-conditioned predictor of the
\emph{response embedding} that each simulation would produce. Let $s_i$ denote
a compact context associated with simulation $i$, and let
$g_\omega$ be a prediction model. We write
$\theta_i(p):=g_\omega(s_i,p)\in\mathbb S^{d-1}$.\footnote{The predictor may, e.g., be a lightweight MLP, or a shared neural network
conditioned on \(s_i\). It can be trained offline by querying the
simulations on prompts and embedding their outputs with $\psi$. At inference
time, computing $\theta_i(p)$ requires only a forward pass of the cheap
predictor.}
Thus, $\theta_i(p)$ predicts neither the text response
itself nor the unavailable preference embedding $\phi_i(p)$, but instead the
embedding of the simulation's response.

The resulting procedure \emph{Predictive route--then--simulate (PRS)} performs one comparatively cheap embedding prediction for every
member of the simulation pool, but makes only $k$ expensive generative calls.
Importantly, no direct access to $(\phi_i)_{i\in \mathcal{N}}$ is assumed in this
procedure.

\begin{PredictRouteSimulatebox}
    
\label{proc:predict-route-simulate}
\raggedright

\textbf{Input:} prompt \(p\),
simulations \(\{A_i\}_{i\in\mathcal{N}}\),
predictor \(g_\omega\), slate size \(k\).

\textbf{Output:} slate \(Y\) of \(k\) simulated viewpoints.

\medskip

\begin{enumerate}[leftmargin=*, itemsep=0.5em]

    \item \textbf{Predict response embeddings.}
    Evaluate the cheap predictor \(g_\omega\) to obtain
    \(
        \theta_i(p)
    \)
    for every \(i\in\mathcal{N}\).

    \item \textbf{Select \(k\) simulations to route to.}
   $R:=\textsc{SimulationRouter}((\theta_i(p))_{i\in \mathcal{N}},k)$

    \item \textbf{Query selected simulations.}
    Query only the simulations in \(R\) and return
    \[
        Y := \{A_r(p): r\in R\}.
    \]

\end{enumerate}
\end{PredictRouteSimulatebox}

Predictive Route--then--Simulate has two sources of error: (1) the simulation fidelity gap $\alpha$ discussed in the previous section and (2) the error of $g_\omega$ in predicting the embeddings of simulations' responses, which we call the 
\emph{VPE prediction error}, defined as
\(
    \eta:=\max_{i\in\mathcal N}
    \angle\!\left(\theta_i(p),\psi(A_i(p))\right).
\)

\begin{restatable}[PRS Approximate Representation]{theorem}{ptsgreward} \label{thm:ptsg-reward} Fix a prompt $p$ and suppose $\operatorname{Simgap}(p)\leq\alpha$ and the VPE prediction error is at most $\eta$. Then predictive Route--then--Simulate returns a slate satisfying $\widetilde\gamma_\beta$-approximate reward PJR+ for \( \widetilde\gamma_\beta=\beta+2\alpha+4\eta. \) \end{restatable}

When $\eta=0$, the predicted embeddings obtain the simulation-response
embeddings exactly, $\theta_i(p)=c_i(p)$, recovering the exhaustive VPE-access
case in \Cref{prop:psipsi}. If the fidelity gap is also zero, then
$c_i(p)=\phi_i(p)$, so this geometry additionally coincides with the
personalized reward geometry used by route--then--simulate.

\begin{figure}[h]
\centering
\[
\begin{array}{rcl}
\text{\textbf{General theory:}}
&
\phi_i(p)
\;\xleftrightarrow{\;\alpha\;}\;
\psi(A_i(p))
\;\xleftrightarrow{\;\eta\;}\;
\theta_i(p)
\\[1.2em]
\text{\textbf{Empirically observable:}}
&
\psi(A_i(p))
\;\xleftrightarrow{\;\eta\;}\;
\theta_i(p)
\end{array}
\]
\caption{
In practice, we can evaluate the VPE prediction error $\eta$ by exhaustively
generating simulation responses and embedding them with $\psi$. The fidelity
gap $\alpha$ depends on the unavailable prompt-conditioned preference embedding
$\phi_i(p)$. When $\alpha=0$, $\psi(A_i(p))=\phi_i(p)$.
}
\label{fig:vpe-observable}
\end{figure}
\paragraph{Observable Viewpoint-Space Evaluation.}
Unlike the fidelity gap $\alpha$, which depends on the unavailable reward
embeddings $\phi_i(p)$, the exhaustive VPE representations can be computed offline
and provide a directly measurable target for empirical evaluation.
We therefore state a guarantee directly in this viewpoint space.

The following corollary transfers representation from the predicted geometry to the exhaustive VPE geometry.

\begin{corollary}[Observable Viewpoint-Space Guarantee]
\label{cor:vpe-observable}
Let
\(\eta:=\max_{i\in\mathcal N}\angle\!\left(\theta_i(p),c_i(p)\right)\).
For every \(j\notin R\) and every \(\ell\)-large group
\(S\subseteq\mathcal N\) such that
\(\max_{i\in S}\angle\!\left(\theta_i(p),\theta_j(p)\right)\le \beta\),
there are at least \(\ell\) indices \(r\in R\) such that, for each such \(r\),
some \(i\in S\) satisfies
\[
\angle\!\left(c_i(p),c_r(p)\right)\le \beta+2\eta.
\]
\end{corollary}

\section{Experiments}\label{sec:exps}
Our goal is to enable a critical component of simulation-augmented generation: routing simulation  queries to simulations in a way that is both \emph{scalable} and \emph{representative}. Thus far, we have developed the necessary theoretical components. In this section, we now empirically measure the level of representation (i.e., the mPJR+ satisfaction rate) attained by our proposed method \emph{Predictive Route-then-Simulate}. 

\subsection{Experimental set-up}
Since we do not have access to reward functions, in our experiments we evaluate how well responses selected by PRS represent the viewpoints expressed by the simulation pool. This lets us evaluate the role of prediction error in PRS empirically \footnote{Their relation to human preferences depends additionally on the simulation fidelity gap.}.
In the language of the previous section, this means for each prompt \(p\), we treat the generated response embeddings \(c_i(p)=\psi(A_i(p))\) as if they were the true preference embedding of human $i$, corresponding to \(v_i=c_i(p)\) and \(\alpha=0\). Our proportional representation metric, mPJR+ relative to the simulation pool, is thus evaluated using the generated embeddings as both voters and candidates.

\textbf{Methods.} We compare the mPJR+ satisfaction rate of the following four methods.

\begin{enumerate}
    \item \emph{Simulate-then-Route (oracle).} This oracle method, which queries all $n$ simulations, always satisfies mPJR+.  It runs SEAR using the generated response embeddings \(\{\psi(A_i(p))\}\) as candidates and voters.
    \item \emph{Predictive Route-then-Simulate (ours).} Our proposed method runs SEAR with the proxy embeddings $\{\theta_i(p)\}$ as both the voters and candidates to select which $k$ simulations to query.
    \item \emph{Predictive Route-then-Simulate, with $k$-means in place of SEAR.} This baseline method runs \(k\)-means on the proxy embeddings \(\{\theta_i(p)\}\) and selects the \(k\) simulations closest to the cluster centers.
    \item \emph{Random.} This baseline selects $k$ simulations uniformly at random.
\end{enumerate}

\paragraph{Domains.} We consider the following two domains, which are representative of two of the highest-impact domains for the deployment of SAGE: (i) contentious political issues, and (ii) personal advice.

\begin{enumerate}
    \item \emph{Simulation of open-ended viewpoints on Remesh.} Remesh~\citep{konya2023democratic,konya2025} is a collective response platform where participants share their open-ended viewpoints on contentious questions and vote on whether they agree or disagree with other participants' viewpoints. We specifically focus on an open-source Remesh session in which $n=310$ representative participants\footnote{Participants were recruited via stratified sampling to be demographically representative of the general U.S. population. While not the focus of these experiments, our stratified-sampling results in \Cref{app:sampling} also bound how representation guarantees extend beyond the \(n\) participants to the U.S. population more broadly.} shared their viewpoints regarding the U.S. campus protests of 2024.\footnote{The data are available at \url{https://github.com/akonya/polarized-issues-data} under a CC-By-4.0 license.}  Because the Remesh dataset only contains four questions, we generate 1{,}000 new prompts related to the original questions to use for training and evaluation. See \Cref{app:remesh-dataset} for details.
    \item \emph{Simulation of responses on Reddit's ``Am I the Asshole?'' (AITA) subreddit.} On the Reddit AITA subreddit, users post personal scenarios and seek guidance on whether they were in the wrong in the situation. Specifically, we use the dataset of AITA posts and responses collected by \citet{wu2026humanlm} (``\texttt{humanual-opinion}'') which comprises $992$ posts and $n = 4567$ users.
\end{enumerate}

\paragraph{Simulators.}  We test two different types of simulators, each of which adapts to an individual through what user context is provided in the system prompt.

\begin{enumerate}
    \item \emph{Remesh.} For Remesh, we construct one simulation per participant. Each participant's context contains their demographic profile\footnote{Gender, age, political affiliation, religious affiliation, education, urbanicity, and household income.}, along with their written viewpoints and binary agree/disagree votes on other participants’ opinions as few-shot context (see \Cref{app:agent-prompt} for the full prompt). The simulator is based on Llama-3.1-8B-Instruct~\citep{grattafiori2024llama} and simply provides the user context as a system prompt with no additional fine-tuning.
    \item \emph{Reddit.} For Reddit, we use the HumanLM simulator which was trained by \citet{wu2026humanlm} on the \texttt{humanual-opinion} dataset. The user context is a summary of the user's Reddit history. The simulator was trained via RL to take in the user context and generate responses that are aligned with latent states (e.g., emotion and stance) of the user's ground-truth responses.
\end{enumerate}

\textbf{The viewpoint embedding model $\psi$ and proxy $\theta$.} For the viewpoint embedding function \(\psi\), we use the model trained by \citet{prefembeddings2026}. The embedding model is fine-tuned on several human preference datasets so that distances in \(\psi\)-space predict preferences rather than merely reflecting surface-level similarity, specifically, so that a human \(i\) is more likely to agree with responses that are closer to their own embedded response. To train the cheap proxy embedding model $\theta$, we fine-tune a copy of the $\psi$ embedding model with LoRA to take in the user context and prompt and output the ground-truth response embedding. See training details in \Cref{app:training-details}.

\subsection{Results}

\textbf{Quantitative results.} For both datasets, \Cref{fig:mpjrp-results} shows the percentage of test prompts for which each method satisfies mPJR+.\footnote{To verify mPJR+, we use the fixed-parameter tractable verification algorithm given by \citet{he2026checkpleaseverifiablyfair} (Algorithm 1, \citet{he2026checkpleaseverifiablyfair}). The algorithm has a runtime of $O(mn\log n \cdot 2^k)$. The exponential dependence on $k$ can make verification slow, especially when $n$ is also large, so for evaluation on Reddit, we subsample $300$ users out of the $n=4567$ total users in the dataset.} In both domains, our method (Route-then-Simulate, SEAR) outperforms the two other Route-then-Simulate methods ($k$-means and random) for all $k \in \{3, \dots, 10\}$. However, the performance of all methods degrades as a function of $k$, in contrast to the oracle Simulate-then-Route method which always gets a 100\% satisfaction rate. We hypothesize that performance degrades because higher values of $k$ require representing smaller groups (of size at least $n/k$), and the representation of these smaller groups may be more sensitive to prediction error between the proxy and ground-truth viewpoint embeddings.

\begin{figure}
    \centering
    \includegraphics[width=0.95\linewidth]{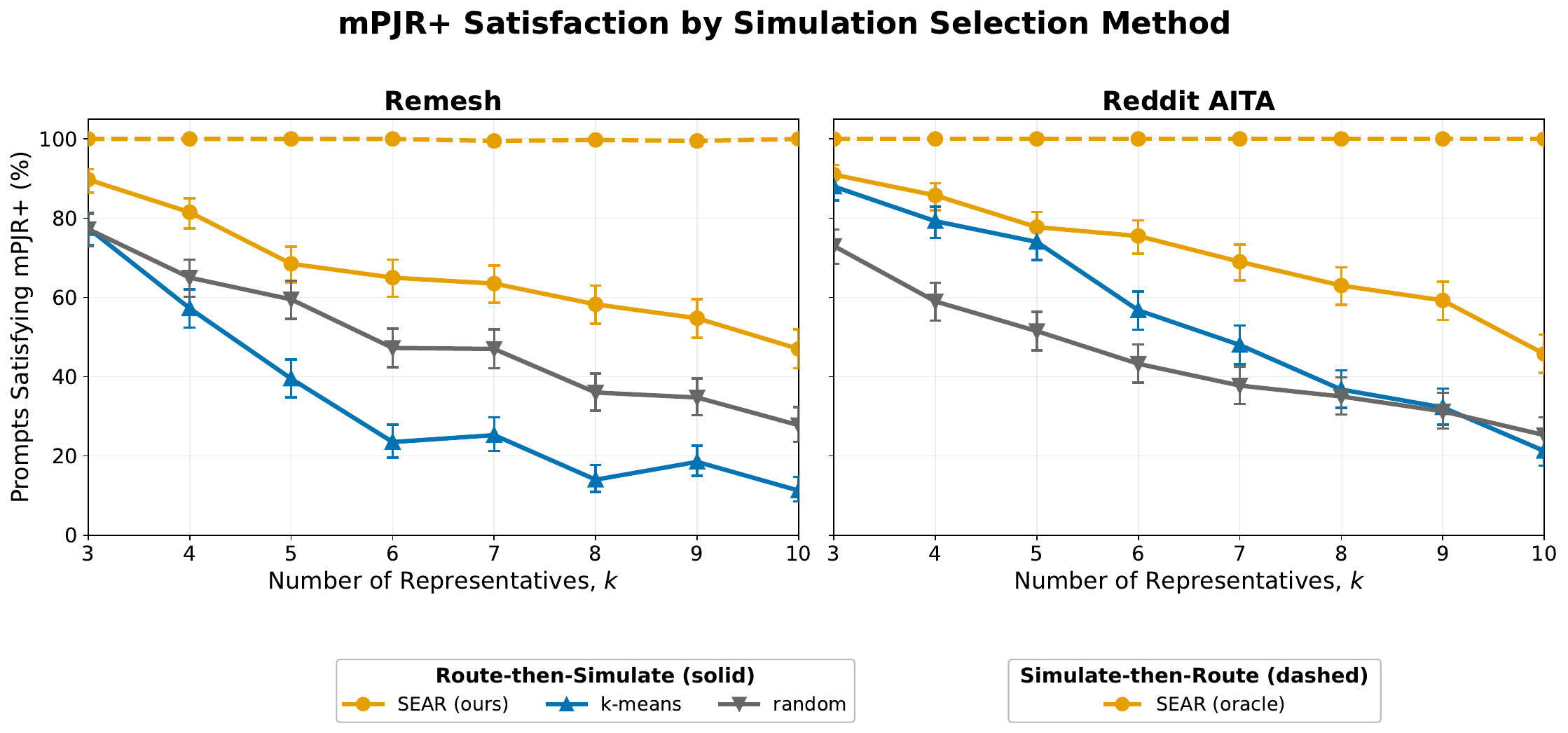}
    \caption{\textbf{mPJR+ Satisfaction Rate by Simulation Selection Method.}} 
    \label{fig:mpjrp-results}
\end{figure}

\textbf{Qualitative example.} Figure~\ref{fig:remesh-rank04-pca} shows a qualitative example from Remesh for the prompt, ``Should pension fund managers be legally allowed to consider ethical or political factors when making investment decisions, even if it might slightly reduce returns?" The figure shows the viewpoint embeddings of the $n$ simulated responses (gray) projected onto the first two PCA dimensions. For each method, the figure highlights which $k$ responses are selected by the method (orange), and if mPJR+ is violated, which users lack adequate representation (pink). 

In this example, Route-then-Simulate SEAR run on the proxy embeddings $\{\theta_i(p)\}$ satisfies mPJR+ and selects very similar responses to those of Oracle SEAR run on the ground-truth viewpoint embeddings $\psi(A_i(p))$. In contrast, Route-then-Simulate $k$-means and Random violate mPJR+. The PCA projection helps explain why: PC1 accounts for 62\% of the variance, while PC2 accounts for only 5\%, and about 80\% of the population lies on the negative side of PC1, corresponding to agreement that pension fund managers should be allowed to consider ethical or political factors. 

When selecting $k=4$ representatives, the groups that are large enough to warrant representation via proportionality must have size at least $1/4$ $(25\%)$ of the population. The disagreement cluster is smaller than this threshold, so it is not entitled to guaranteed representation. Nonetheless, \(k\)-means and Random select outliers with positive PC1 values (corresponding to disagreement). This leaves the dense agreement cluster underrepresented in the \(k=4\) selected responses, causing mPJR+ to fail. \Cref{tab:remesh-rank04-pca-slates} reports the text of the $k$ responses that each method selects. More qualitative examples for Remesh and Reddit are available in \Cref{app:qual-examples}.

\begin{figure*}[t]
    \centering
    \includegraphics[width=\textwidth]{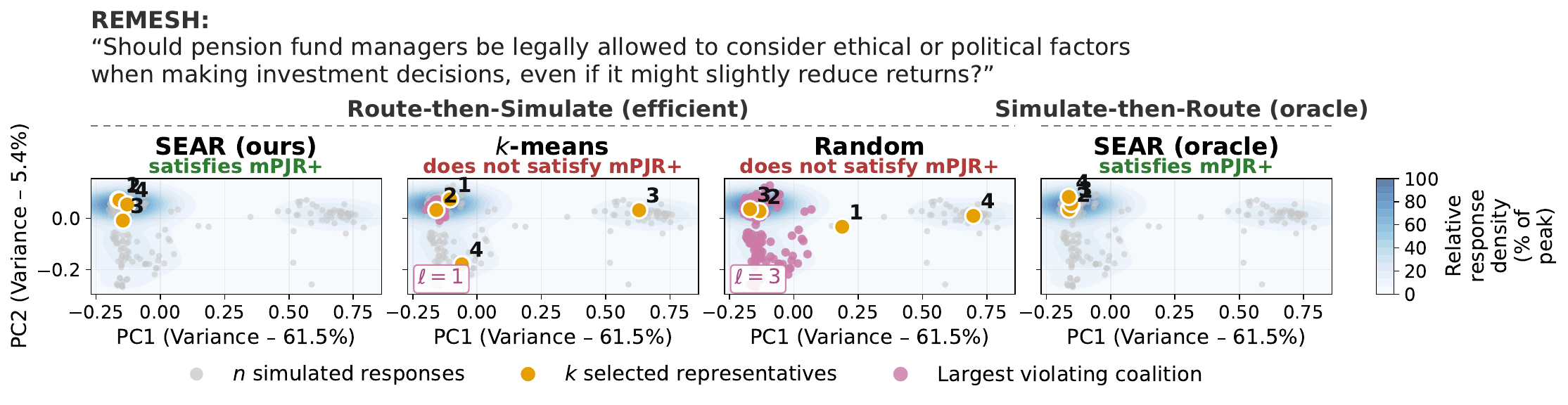}
    \caption{\textbf{Qualitative example on Remesh.} Response embeddings and selected representatives for the prompt ``Should pension fund managers be legally allowed to consider ethical or political factors when making investment decisions, even if it might slightly reduce returns?" For each method, \Cref{tab:remesh-rank04-pca-slates} reports the text of its $k=4$ selected responses.}
    \label{fig:remesh-rank04-pca}
\end{figure*}


\section{Discussion} \label{sec:disc}
In this work, we provided the first formalization of scalable and representative routing for simulation-augmented generation, grounded in social choice theory. We showed that to represent a population of $n_H$ humans, we need only create $n \ll n_H$ simulations of them, and need only dynamically query $k \ll n$ of those simulations at inference time, while still maintaining approximate proportional representation guarantees for the full population.

\textbf{Limitations.} Several limitations warrant discussion. Our empirical evaluation uses two datasets, leaving generalization to other domains and populations as future work. To systematically study the impact of simulation fidelity gap $\alpha$, it remains to be seen what levels of $\gamma$ are expected in the wild, although we generally expect fidelity to improve as simulation methods improve. Finally, as is standard in work on committee elections, the choice of $k$ is left open and must be tailored to the application context, since it directly determines the threshold group size required for representation. 

\textbf{Future work.} Several directions remain open. First, future work could further push the Pareto frontier between scalability and representativeness, including exploring other axioms beyond reward PJR+. Second, in practice, different simulations may be more reliable in different domains (e.g., politics versus interpersonal advice). How should routing account for heterogeneous fidelity, e.g., if certain subgroups of the population have systematically different fidelity? 

\textbf{Broader impacts.} Our work is situated within the broader context of simulation-augmented generation (SAGE). That said, the specific contribution of this paper addresses only one component of this vision: the routing problem of efficiently selecting representative simulations at inference time. A trustworthy implementation of SAGE requires several additional components—including high-fidelity simulation, confidence estimation, and appropriate user interfaces. We caution against deploying the routing methods developed here in isolation, without adequate attention to these complementary elements.

\newpage
\bibliographystyle{unsrtnat}
\bibliography{references}

\newpage
\appendix

   \clearpage

\section{Additional Material on \Cref{sec:reward-commutation}}

In this section, we provide additional details on and proofs omitted from \Cref{sec:reward-commutation}.

\subsection{Proofs Omitted from \Cref{sec:simulate-then-route}}
\gtsitprf*
\begin{proof}
Let \(C=(o_i)_{i\in\mathcal N}\) be the responses generated in the first
step of Simulate--then--Route. The second step applies
\(\textsc{ResponseRouter}\) to \(C\).

Fix \(\ell\in[k]\), an unselected response \(c\in C\setminus Y\), and an
\(\ell\)-large group \(S\subseteq\mathcal N\). Let
\[
    \tau:=\min_{i\in S} r_i(p,c).
\]
By threshold \(\tau\), every human in \(S\) approves \(c\). Since \(c\) is not
selected, the group \(S\) must have spent enough of its budget on previously
selected responses that it can no longer fund \(c\). Since
\(|S|\ge \ell n/k\), this implies that at least \(\ell\) selected responses
received positive budget from members of \(S\) at reward thresholds at least
\(\tau\). For every such selected response \(o\), there is some \(i\in S\)
such that
\[
    r_i(p,o)\ge \tau
    =\min_{j\in S}r_j(p,c).
\]
Hence at least \(\ell\) responses \(o\in Y\) satisfy
\[
    \max_{i\in S}r_i(p,o)
    \ge \min_{i\in S}r_i(p,c),
\]
so \(Y\) satisfies reward PJR+.
\end{proof}

\paragraph{Fast implementation of REAR and SEAR.}
We prove Theorem~2 using the common bookkeeping underlying REAR and SEAR,
and present the implementation below in the geometric SEAR notation for
convenience. For REAR, one simply orders humans for each candidate by
decreasing reward rather than increasing distance, and replaces the
minimum-radius threshold by the corresponding maximum reward threshold;
the budget updates and runtime analysis are otherwise identical. Thus,
the analysis below establishes the claimed $O(mn\log n)$ runtime for
REAR, while the corresponding runtime bound for SEAR follows in
Corollary~1.

For comparison, \citet{prf} analyze SEAR in the special case in which
the same $n$ points serve as both voters and candidates. Their proposed
implementation and accompanying analysis give the following
polynomial-time bound\footnote{Versions of this paper have been presented at several workshops, including EEAMO '25, IASEAI '26, and ICML '26. The cited result refers to the version of the paper available at the time we established our result. In July 2026, the authors updated their paper with a runtime of \(O(n^2(\log n + k))\).}.

\begin{theorem}[\citet{prf}]
\label{thm:slow}
In the setting with $n$ voters and the same $n$ points as candidates,
SEAR terminates in time $O(n^4k^2)$.
\end{theorem}

Our application uses the more general formulation with a voter multiset
$V=(v_1,\ldots,v_n)$ and a possibly different candidate multiset
$C=(c_1,\ldots,c_m)$. We treat voters and candidates as indexed
objects, so distinct voters or candidates may occupy the same point.
We assume, as usual, that $k\le \min\{m,n\}$.

Each voter $q\in[n]$ begins with budget $b_q=1$. In each round, every
unselected candidate expands a ball around itself until the voters in
the ball have total remaining budget at least $n/k$. The candidate
whose ball first reaches this threshold is selected, and its approving
voters pay a total of $n/k$. Ties between candidates and between
equidistant voters are resolved by index order.

\begin{algorithm} \caption{Spatial Expanding Approval Rule (SEAR)} \label{alg:sear-fast} \begin{algorithmic}[1] \Require Voters $V=(v_1,\ldots,v_n)$, candidates $C=(c_1,\ldots,c_m)$, slate size $k$ \State $X\gets ()$ \State $b_q\gets 1$ for every $q\in[n]$ \For{each $p\in[m]$} \State Let $L_p=(L_p[1],\ldots,L_p[n])$ be the voters ordered by increasing $\Delta(c_p,v_q)$, breaking ties by voter label \EndFor \For{$t=1,\ldots,k$} \State $A\gets\{p\in[m]:c_p\notin X\}$ \For{each $p\in A$} \State Let $i(p)$ be the smallest index such that \[ \sum_{j=1}^{i(p)} b_{L_p[j]} \ge \frac nk . \] \State $r_p\gets \Delta(c_p,v_{L_p[i(p)]})$ \EndFor \State Choose $p^*\in\arg\min_{p\in A}r_p$, breaking ties by label order \State Append $c_{p^*}$ to $X$ \State $\rho\gets n/k$ \Comment{remaining amount to collect from the selected ball} \For{$j=1,\ldots,i(p^*)$} \State $q\gets L_{p^*}[j]$ \State $x_q\gets\min\{b_q,\rho\}$ \Comment{amount paid by voter $q$} \State $b_q\gets b_q-x_q$ \Comment{update $q$'s remaining budget} \State $\rho\gets\rho-x_q$ \Comment{update remaining joint payment} \If{$\rho=0$} \State \textbf{break} \EndIf \EndFor \EndFor \State \Return $X$ \end{algorithmic} \end{algorithm}

The following theorem improves the $O(n^4k^2)$ bound and applies to the
general setting with $m$ candidates and $n$ voters. In particular, it
has no separate multiplicative dependence on $k$.

\searruntime*

\begin{proof}
We analyze the implementation in \Cref{alg:sear-fast}. First compute
all $mn$ candidate--voter distances. For every candidate $c_j$, sort
the $n$ voter labels by increasing distance from $c_j$, breaking ties
by voter index, and store both the resulting list $L_j$ and its inverse
rank map. This preprocessing takes $O(mn\log n)$ time.

The algorithm maintains the remaining voter budgets
$b_1,\ldots,b_n$. For each unselected candidate $c_j$, it also
maintains an index $i(j)$ and the partial sum
$B(j):=\sum_{h=1}^{i(j)}b_{L_j[h]}$. At the beginning of each round,
$i(j)$ is advanced until $B(j)\ge n/k$. Because $L_j$ is ordered by
distance from $c_j$, the quantity
$r_j:=\Delta(c_j,v_{L_j[i(j)]})$ is the smallest radius at which the
ball around $c_j$ contains total remaining budget at least $n/k$.

The algorithm selects an unselected candidate $c_{j^*}$ minimizing
$r_j$. It then charges the voters in the corresponding prefix of
$L_{j^*}$ in order until they have paid $n/k$ in total. Every
positive-budget voter before the final contributor spends its entire
remaining budget, while the final contributor pays exactly the
remaining amount. Since $i(j^*)$ is minimal, all charged voters lie in
the selected candidate's approval ball at radius $r_{j^*}$. Thus this is exactly a valid execution of SEAR under the stated
tie-breaking and payment rules. The corresponding statement for REAR
follows by replacing increasing distance with decreasing reward, as
described above.

We now bound the cost of maintaining the partial sums. Suppose voter
$q$ pays an amount $x_q>0$. For every unselected candidate $c_j$, voter
$q$ appears in the maintained prefix precisely when
$\operatorname{rank}_j(q)\le i(j)$. We can test this condition in
constant time and, when it holds, subtract $x_q$ from $B(j)$. Hence,
processing one positive contribution costs $O(m)$.

Every contributor other than the final contributor in a round spends
its entire remaining budget and can therefore be such a contributor at
most once during the execution. There are consequently at most $n$
positive non-final contributions. There is at most one final
contributor per round, giving at most $k$ further positive
contributions. Thus the total number of positive contributions is at
most $n+k\le 2n$, and all partial-sum updates take $O(mn)$ time.

After budgets decrease, some partial sums may fall below $n/k$. For
each fixed candidate $c_j$, the index $i(j)$ only increases and never
exceeds $n$. The total number of index increments over the entire
execution is therefore at most $mn$, so all threshold recomputations
take $O(mn)$ time.

Finding the minimum-radius unselected candidate by scanning the
candidate set costs $O(m)$ per round and $O(mk)\le O(mn)$ in total.
Collecting a payment scans at most $n$ voter positions per round, for
a total of $O(nk)\le O(nm)$. All work after preprocessing therefore
takes $O(mn)$ time.

Combining this with the $O(mn\log n)$ preprocessing cost, the total
amortized running time is
$O(mn\log n)+O(mn)=O(mn\log n)$.
\end{proof}
\subsection{Details Omitted from \Cref{sec:geom-rep}}\label{app:geom-rep}
 The following is an example demonstrating that k-means fails mPJR+.

\begin{example}
Consider six indexed voters and candidates in $\mathbb R^2$:
\[
v_1=v_2=v_3=v_4=(0,0),\qquad
v_5=(1,0),\qquad
v_6=(0,1),
\]
and let $k=3$. An optimal $k$-means solution places one center at each
of the three occupied locations and has objective value zero. Selecting
one candidate closest to each center therefore gives, up to
tie-breaking, the slate
$X=\{v_1,v_5,v_6\}$.

Let $S=\{v_1,v_2,v_3,v_4\}$ and $\ell=2$. Since
$|S|=4=2n/k$, the group $S$ is $\ell$-large. Consider the unselected
candidate $c=v_2$. Since all members of $S$ and $c$ are located at the
origin,
$\max_{i\in S}\Delta(v_i,c)=0$. Hence mPJR+ requires at least two
members $x\in X$ satisfying
$\min_{i\in S}\Delta(v_i,x)\le 0$. Only $v_1$ satisfies this
inequality, so the $k$-means slate violates mPJR+.
\end{example}

Moreover, balanced k-means also fails mPJR+.
\begin{example}
Balanced $k$-means does not necessarily satisfy mPJR+, even in
$\mathbb R^2$. Let $n=9$ and $k=3$, and define
\[
O=(0,0),\qquad
A=(1,0),\qquad
B=\left(-\frac12,\frac{\sqrt3}{2}\right),\qquad
C=\left(-\frac12,-\frac{\sqrt3}{2}\right).
\]
Consider three indexed voters and candidates at $O$ and two at each of
$A$, $B$, and $C$.

Balanced $k$-means partitions the points into three clusters of size
three. An optimal balanced clustering is
\[
\{O,A,A\},\qquad \{O,B,B\},\qquad \{O,C,C\}.
\]
Moreover, every optimal balanced clustering has this form, up to the
labels of coincident points. Indeed, for a cluster of three points, the
within-cluster sum of squares is one third of the sum of its pairwise
squared distances. The displayed clustering has total cost $2$. If the
three copies of $O$ are distributed among the clusters as $(2,1,0)$,
the cost is at least $10/3$, while the distribution $(3,0,0)$ has cost
at least $4$. Under the remaining distribution $(1,1,1)$, cost $2$ is
attained only when the two non-origin points in each cluster coincide.

The three cluster centroids are $2A/3$, $2B/3$, and $2C/3$. Selecting
a candidate nearest each centroid therefore selects one candidate at
each of $A$, $B$, and $C$, and no candidate at $O$.

Let $S$ be the group consisting of the three candidates at $O$, and
let $\ell=1$. Since $|S|=3=n/k$, the group is $\ell$-large. For an
unselected candidate $c$ at $O$, we have
$\max_{i\in S}\Delta(v_i,c)=0$. However, no selected candidate $x$
satisfies $\min_{i\in S}\Delta(v_i,x)\le 0$. Thus the balanced
$k$-means slate violates mPJR+.
\end{example}

\sampledpop*

\begin{proof}
It suffices to consider $\ell$ such that $\ell/k+\varepsilon\le 1$,
since otherwise no such group $S$ exists. For each such $\ell$, let
\(q_\ell
    :=
    \left\lceil(\ell/k+\varepsilon)n_H\right\rceil.\)
Write \(v_i:=\phi_i(p)=\psi(o_i^*(p)),\)
and let $\angle(v_i,v_j)$ denote angular distance. For each
$j\in\mathcal H$, let $S_j^\ell$ consist of the $q_\ell$ population
indices closest to $v_j$ in angular distance, with ties resolved by a
fixed rule independent of the sample, and set
$T_j^\ell:=S_j^\ell\cap\mathcal N.$
For every fixed pair $(j,\ell)$,
$\mathbb E\left[\frac{|T_j^\ell|}{n}\right]
    =
    \frac{q_\ell}{n_H}
    \ge
    \frac{\ell}{k}+\varepsilon.$
Hence Hoeffding's inequality for sampling without replacement gives
\[
    \Pr\left(
        |T_j^\ell|<\frac{\ell n}{k}
    \right)
    \le
    2e^{-2n\varepsilon^2}.
\]
Taking a union bound over at most $kn_H$ pairs $(j,\ell)$,
\[
    \Pr\left(
        \exists j,\ell:
        |T_j^\ell|<\frac{\ell n}{k}
    \right)
    \le
    2kn_H e^{-2n\varepsilon^2}
    \le\delta,
\]
by the choice of $n$. We condition on the complementary event for the
remainder of the proof.

Fix a size-$k$ slate
$Y\subseteq\mathcal C_{\mathcal N}$ satisfying rPJR+ on the sampled
instance. Fix $\ell\in[k]$, a comparison response
$c=o_j^*(p)\in\mathcal C_{\mathcal H}\setminus Y,$
and a group $S\subseteq\mathcal H$ with
$
    |S|\ge(\ell/k+\varepsilon)n_H,
    r_i(p,c)\ge\cos\beta$
    for all $i\in S.
$
If $\beta\ge\pi/4$, then $\gamma_\beta=\pi$ and the claim is trivial,
so suppose $\beta<\pi/4$.

By Assumption~1 and the definition of $o_j^*(p)$,
$r_i(p,c)
    =
    v_i^\top v_j
    =
    \cos\angle(v_i,v_j).$
Since cosine is decreasing on $[0,\pi]$,
$\angle(v_i,v_j)\le\beta$ for every $i\in S$.
Thus at least $q_\ell$ population members lie within angular distance
$\beta$ of $v_j$. By the definition of $S_j^\ell$,
$\angle(v_t,v_j)\le\beta$ for every $t\in S_j^\ell.$ Set $T:=T_j^\ell$. By the concentration event,
$
    |T|\ge\frac{\ell n}{k}.
$

If at least $\ell$ indices in $T$ have their corresponding responses
selected in $Y$, then for each such response $o_x^*(p)\in Y$ and any
$i\in S$,
$
    \angle(v_i,v_x)
    \le
    \angle(v_i,v_j)+\angle(v_j,v_x)
    \le2\beta,
$
so
$
    r_i(p,o_x^*(p))
    \ge\cos(2\beta),
$
and the claim follows. Otherwise, fewer than $\ell$ indices of $T$ correspond to selected
responses. Since $n\ge k$,
$|T|\ge\frac{\ell n}{k}\ge\ell,$
so there exists some $a\in T$ such that
$o_a^*(p)\notin Y$.

For all $t\in T$,
$\angle(v_t,v_a)
    \le
    \angle(v_t,v_j)+\angle(v_j,v_a)
    \le2\beta.$
Hence
$r_t(p,o_a^*(p))
    =
    \cos\angle(v_t,v_a)
    \ge\cos(2\beta)
    \qquad\text{for every }t\in T.$
Because $T$ is $\ell$-large in the sampled population and
$o_a^*(p)\in\mathcal C_{\mathcal N}\setminus Y$, sampled rPJR+ implies
that at least $\ell$ responses $o=o_x^*(p)\in Y$ satisfy
$\max_{t\in T} r_t(p,o)
    \ge
    \min_{t\in T}r_t(p,o_a^*(p))
    \ge
    \cos(2\beta).$
For each such $o=o_x^*(p)$, there therefore exists $t_x\in T$ such that
$\angle(v_{t_x},v_x)\le2\beta.$

Now fix any $i\in S$. Since $t_x\in T\subseteq S_j^\ell$,
\[
\begin{aligned}
    \angle(v_i,v_x)
    &\le
    \angle(v_i,v_j)
    +\angle(v_j,v_{t_x})
    +\angle(v_{t_x},v_x)\\
    &\le
    \beta+\beta+2\beta\\
    &=4\beta.
\end{aligned}
\]
Thus $r_i(p,o)
    =
    \cos\angle(v_i,v_x)
    \ge
    \cos(4\beta).$
Consequently, for each of these at least $\ell$ responses,
\(
    \max_{i\in S}r_i(p,o)\ge\cos(4\beta).
\)
\end{proof}

\rewardmetricequivalence*
\begin{proof}
Fix the prompt $p$, and write $v_i:=\phi_i(p)$ and
$c_j:=\psi(o_j)$. Let $X$ be the indexed embedding slate corresponding
to the response slate $Y$.

Since $v_i$ and $c_j$ are unit vectors,
$r_i(p,o_j)=v_i^\top c_j
=1-\frac{1}{2}\Delta(v_i,c_j)^2$. Thus, for any selected response
$o_a$, any unselected response $o_j$, and any group $S$,
\[
\max_{i\in S}r_i(p,o_a)\ge \min_{i\in S}r_i(p,o_j)
\quad\Longleftrightarrow\quad
\min_{i\in S}\Delta(v_i,c_a)
\le \max_{i\in S}\Delta(v_i,c_j).
\]
Therefore, at least $\ell$ selected responses satisfy the reward PJR+
inequality if and only if their corresponding indexed candidates
satisfy the mPJR+ inequality. Since selected and unselected responses
correspond index-by-index to selected and unselected candidates, $Y$
satisfies reward PJR+ if and only if $X$ satisfies mPJR+.
\end{proof}

\subsection{Proofs Omitted from \Cref{sec:route-then-simulate}}
\exactstg*
\begin{proof}
Write $v_i:=\phi_i(p)$ and $c_i:=\psi(A_i(p))$. By
Assumption~\ref{assump:factorized-reward}, the maximum possible reward
for human $i$ is $1$. Since $A_i(p)$ is reward-maximizing,
$v_i^\top c_i=1$, and hence $c_i=v_i$ because both vectors are unit
vectors.

Therefore, Simulate--then--Route and Route--then--Simulate both run SEAR
using the same indexed multiset $(v_i)_{i\in\mathcal N}$ as voters and
candidates. Under the common tie-breaking rule, they select the same
simulation indices. The final claim follows from
Theorem~\ref{thm:gts-it-prf}.
\end{proof}

\approxstg*
\begin{proof}
Write $v_i:=\phi_i(p)$, $o_i:=A_i(p)$, and
$z_i:=\psi(o_i)$ for each $i\in\mathcal N$. Let $R$ be the set of
indices selected by $SimulationRouter$, let
$X:=(v_r)_{r\in R}$ be the corresponding indexed embedding slate, and
let $Y:=(o_r)_{r\in R}$ be the slate returned by
Route--then--Simulate.

For unit vectors $x,y$, write $\angle(x,y)$ for their angular distance.
Since $x^\top y=\cos(\angle(x,y))$, the condition
$r_i(p,o)\ge\cos(\theta)$ is equivalent to
$\angle(v_i,\psi(o))\le\theta$, for $\theta\in[0,\pi]$. Moreover, the
assumption $Simgap(p)\le\alpha$ gives
$\angle(v_i,z_i)\le\alpha$ for every $i\in\mathcal N$. The slate $X$ satisfies mPJR+ with respect to
$(V,V)$, where $V:=(v_i)_{i\in\mathcal N}$. Although mPJR+ was stated
using Euclidean distance, we may equivalently apply it using angular
distance: for unit vectors,
$\lVert x-y\rVert_2=2\sin(\angle(x,y)/2)$, which is increasing in
$\angle(x,y)$ on $[0,\pi]$. Fix $\ell\in[k]$, an indexed unselected response
$o_j\in C\setminus Y$, equivalently $j\notin R$, and an $\ell$-large
group $S\subseteq\mathcal N$. Suppose that
$r_i(p,o_j)\ge\cos(\beta)$ for every $i\in S$. Equivalently,
$\angle(v_i,z_j)\le\beta$ for every $i\in S$. We first obtain a bound using the unselected candidate $v_j$. For every
$i\in S$, the angular triangle inequality gives
$\angle(v_i,v_j)\le\angle(v_i,z_j)+\angle(z_j,v_j)
\le\beta+\alpha$. Since $j\notin R$, the indexed candidate $v_j$ is
unselected from $X$. Applying mPJR+ to the group $S$ and candidate
$v_j$, there are at least $\ell$ selected indices $r\in R$ such that
$\min_{i\in S}\angle(v_i,v_r)
\le\max_{i\in S}\angle(v_i,v_j)\le\beta+\alpha$.

For each such $r$, choose $i_r\in S$ with
$\angle(v_{i_r},v_r)\le\beta+\alpha$. After generation,
$\angle(v_r,z_r)\le\alpha$, and hence
$\angle(v_{i_r},z_r)\le\beta+2\alpha$. Thus, at least $\ell$ responses
$o_r\in Y$ receive reward at least
$\cos(\min\{\pi,\beta+2\alpha\})$ from some member of $S$. We next obtain a second bound from the diameter of $S$. For every
$i,i'\in S$, both $v_i$ and $v_{i'}$ are within angle $\beta$ of
$z_j$, so $\angle(v_i,v_{i'})\le 2\beta$. If $|R\cap S|\ge\ell$, then the corresponding selected embeddings
already lie in $S$, and after generation their responses are within
angle $\alpha$, and therefore within angle $\alpha+2\beta$, of a member
of $S$.

Otherwise, $|R\cap S|<\ell$. Since $S$ is $\ell$-large and $n\ge k$,
we have $|S|\ge\ell n/k\ge\ell$, so there is some $s\in S\setminus R$.
Applying mPJR+ to $S$ and the unselected candidate $v_s$, there are at
least $\ell$ selected indices $r\in R$ such that
$\min_{i\in S}\angle(v_i,v_r)
\le\max_{i\in S}\angle(v_i,v_s)\le 2\beta$.
For each such $r$, generation adds at most $\alpha$ to this distance,
so $\angle(v_i,z_r)\le\alpha+2\beta$ for some $i\in S$. Thus, in either case, at least $\ell$ responses in $Y$ are within angle
$\alpha+2\beta$ of some member of $S$. Combining this with the first
bound and choosing the better of the two arguments, at least $\ell$
responses $o_r\in Y$ satisfy, for some $i\in S$,
\begin{align*}
\angle(v_i,z_r)
&\le
\min\{\pi,\beta+2\alpha,\alpha+2\beta\} \\
&=
\min\{\pi,\alpha+\beta+\min\{\alpha,\beta\}\}
=:\gamma_\beta.
\end{align*}
Consequently, each such response satisfies
$r_i(p,o_r)\ge\cos(\gamma_\beta)$, and therefore
$\max_{i\in S}r_i(p,o_r)\ge\cos(\gamma_\beta)$.

Since $\ell$, the unselected indexed response $o_j$, and the
$\ell$-large group $S$ were arbitrary, $Y$ satisfies
$\gamma_\beta$-approximate reward PJR+.
\end{proof}

\subsection{Proof Omitted from \Cref{sec:predicted-response-embeddings} }
\ptsgreward*
\begin{proof}
For each $i\in\mathcal N$, write $v_i:=\phi_i(p)$,
$c_i:=\psi(A_i(p))$, and $\theta_i:=\theta_i(p)$. Let
$R:=SimulationRouter((\theta_i)_{i\in\mathcal N},k)$ be the set of
selected indices, let $X:=(\theta_r)_{r\in R}$ be the corresponding
indexed predicted-embedding slate, and let
$Y:=(A_r(p))_{r\in R}$ be the indexed response slate returned by
Predictive Route--then--Simulate.

For unit vectors $x,y$, write $\angle(x,y)$ for their angular distance.
The assumptions $\operatorname{Simgap}(p)\le\alpha$ and VPE prediction
error at most $\eta$ imply
$\angle(v_i,c_i)\le\alpha$ and
$\angle(\theta_i,c_i)\le\eta$ for every $i\in\mathcal N$.

Since $R$ is returned by $SimulationRouter$ on the predicted
embeddings, $X$ satisfies mPJR+ with respect to the voter and candidate
multisets $(\Theta,\Theta)$, where
$\Theta:=(\theta_i)_{i\in\mathcal N}$. Although mPJR+ is stated using
Euclidean distance, it can equivalently be applied using angular
distance, since for unit vectors
$\lVert x-y\rVert_2=2\sin(\angle(x,y)/2)$, which is increasing in
$\angle(x,y)$ on $[0,\pi]$.

Fix $\ell\in[k]$, an indexed unselected response
$o_j:=A_j(p)$ with $j\notin R$, and an $\ell$-large group
$S\subseteq\mathcal N$. Suppose that
$r_i(p,o_j)\ge\cos(\beta)$ for every $i\in S$. Since
$r_i(p,o_j)=v_i^\top c_j$, this is equivalent to
$\angle(v_i,c_j)\le\beta$ for every $i\in S$.

For each $i\in S$, the angular triangle inequality gives
\begin{align*}
\angle(\theta_i,\theta_j)
&\le
\angle(\theta_i,c_i)
+\angle(c_i,v_i)
+\angle(v_i,c_j)
+\angle(c_j,\theta_j) \\
&\le
\eta+\alpha+\beta+\eta
=
\beta+\alpha+2\eta.
\end{align*}
Consequently,
$\max_{i\in S}\angle(\theta_i,\theta_j)
\le\beta+\alpha+2\eta$.

The indexed candidate $\theta_j$ is unselected from $X$. Applying
mPJR+ to the group $S$ and the unselected candidate $\theta_j$, there
are at least $\ell$ selected indices $r\in R$ such that
$\min_{i\in S}\angle(\theta_i,\theta_r)
\le\max_{i\in S}\angle(\theta_i,\theta_j)
\le\beta+\alpha+2\eta$.

For each such $r$, choose $i_r\in S$ satisfying
$\angle(\theta_{i_r},\theta_r)\le\beta+\alpha+2\eta$. Using the
simulation-fidelity and prediction-error bounds once more, we obtain
\begin{align*}
\angle(v_{i_r},c_r)
&\le
\angle(v_{i_r},c_{i_r})
+\angle(c_{i_r},\theta_{i_r})
+\angle(\theta_{i_r},\theta_r)
+\angle(\theta_r,c_r) \\
&\le
\alpha+\eta+(\beta+\alpha+2\eta)+\eta \\
&=
\beta+2\alpha+4\eta
=
\widetilde\gamma_\beta.
\end{align*}
It follows that
$r_{i_r}(p,A_r(p))
=v_{i_r}^{\top}c_r
\ge\cos(\widetilde\gamma_\beta)$.
Hence, for each of these at least $\ell$ selected responses,
$\max_{i\in S}r_i(p,A_r(p))
\ge\cos(\widetilde\gamma_\beta)$.

Since $\ell$, the unselected indexed response $o_j$, and the
$\ell$-large group $S$ were arbitrary, the slate returned by
Predict--route--then--simulate satisfies
$\widetilde\gamma_\beta$-approximate reward PJR+.
\end{proof}


\section{Additional Sampling Results}\label{app:sampling}
This section gives two additional population-transfer results for mPJR+ based
on sampling without replacement. First, under uniform sampling, we show that
in Euclidean dimension \(d\), mPJR+ guarantees transfer from a sample to the
full population with sample complexity depending on \(d\), \(\varepsilon\),
and \(\delta\), but not on the population size \(n_H\) or the slate size \(k\).
Second, we extend the population-transfer guarantee of
\Cref{thm:sampled-population} from uniform sampling to stratified sampling via
a conditioning argument.

\paragraph{Dimension-dependent population transfer.}
\begin{lemma}
\label{lem:witness-uniform-convergence}
Let $\varepsilon \in (0, 1/e]$ and $\delta \in (0,1)$.
Let \(\mathcal H\) be a finite population of size \(n_H\), and let
\(\mathcal W\subseteq 2^{\mathcal H}\) be a set system with VC dimension at most
\(h\). Let \(\mathcal N\subseteq\mathcal H\) be a uniformly
random sample of \(n\) agents drawn without replacement. If
\(n=\Omega\!\left(\varepsilon^{-2}(h\log(1/\varepsilon)+\log(1/\delta))\right)\),
then with probability at least \(1-\delta\), every \(W\in\mathcal W\) satisfies
\[
\left|
\frac{|W\cap\mathcal N|}{|\mathcal N|}
-
\frac{|W|}{|\mathcal H|}
\right|
\le \varepsilon .
\]
\end{lemma}
\begin{proof}

For \(W\in\mathcal W\), let \(p(W)=|W|/|\mathcal H|\) and
\(\hat p(W)=|W\cap\mathcal N|/|\mathcal N|\). By Hoeffding's comparison
theorem for sampling without replacement, the usual i.i.d. VC inequality also
upper-bounds the uniform deviation probability for the present sampling model.
Hence, for a sample of size \(n\),
\[
\Pr\left(\sup_{W\in\mathcal W}|\hat p(W)-p(W)|>\varepsilon\right)
\le
8\Pi_{\mathcal W}(2n)\exp(-n\varepsilon^2/32).
\]

Since \(\VCdim(\mathcal W)\le h\), Sauer--Shelah gives
\(\Pi_{\mathcal W}(2n)\le (2en/h)^h\). Hence the failure probability is at most
\[
8(2en/h)^h\exp(-n\varepsilon^2/32).
\]
Thus, it suffices that
\[
n\ge
\frac{32}{\varepsilon^2}
\left(h\log\frac{2en}{h}+\log\frac{8}{\delta}\right).
\]
\begin{claim}
Let \(L:=\log(8/\delta)\), and define
\(f(m):=\frac{32}{\varepsilon^2}\bigl(h\log\frac{2em}{h}+L\bigr)\).
There exists an absolute constant \(C\) such that if
\(m_0:=C\bigl(h\log(1/\varepsilon)+L\bigr)/\varepsilon^2\), then every
\(m\ge m_0\) satisfies \(m\ge f(m)\).
\end{claim}
\begin{proof}
First, $f$ is increasing and concave with derivative
$f'(m)=\frac{32}{\varepsilon^2}\frac{h}{m}$.
Hence for all $m\ge m_\ast:=\frac{64}{\varepsilon^2}h$
we have $f'(m)\le1$.

Choose $m_0:=C\,\bigl(h\log(1/\varepsilon)+L\bigr)/\varepsilon^2$
with $C$ large enough so that $m_0\ge m_\ast$.
We upper bound the logarithm at $m_0$ by
\begin{align*}
\log\!\frac{2em_0}{h}
&\le \log\frac{m_0}{h}+\log(2e)\\
&=\log(2Ce)+2\log(1/\varepsilon)
+\log(\log(1/\varepsilon)+L/h).
\end{align*}

Using that $\log(a+b)\leq\log(a)+b/a$ for $a,b>0$
and setting $M=\log(1/\varepsilon)$,
\begin{align}
\log(2Ce)+2\log(1/\varepsilon)+\log(M+L/h)
\leq \log(2Ce)+3M+\frac{L}{hM}.
\end{align}

Since $\log(1/\varepsilon)>0$ ($\varepsilon<1$) and
$\log(1/\delta)>0$ ($\delta<1$), we obtain
\[
\begin{aligned}
f(m_0)
&\le \frac{32}{\varepsilon^2}
\Bigl(h\bigl[\tfrac{L}{hM}+3M+\log(2Ce)\bigr]+L\Bigr)\\
&= \frac{32}{\varepsilon^2}
\Bigl(\tfrac{L}{M}+3hM+h\log(2Ce)+L\Bigr)\\
&\le \frac{32}{\varepsilon^2}
\Bigl(3hM+h\log(2Ce)+2L\Bigr).
\end{aligned}
\]

Comparing with $m_0=\frac{C}{\varepsilon^2}(hM+L)$,
choosing $C\geq64$ large enough ensures $f(m_0)\le m_0$
for all $h\ge1$ and $L\ge0$.

Finally, for any $m\ge m_0$, concavity with $f'(m_0)\le1$
gives the supporting-line bound
\[
f(m)\le f(m_0)+f'(m_0)(m-m_0)
\le f(m_0)+(m-m_0)\le m.
\]
Thus $m\ge f(m)$ for all $m\ge m_0$, as claimed.
\end{proof}

The claim implies that
\(n=\Omega(\varepsilon^{-2}(h\log(1/\varepsilon)+\log(1/\delta)))\)
is sufficient for the failure probability to be at most \(\delta\). This proves
the lemma.
\end{proof}

Given a finite population \(\mathcal H\subseteq\mathcal X\) and a range family
\(\mathcal R\subseteq 2^{\mathcal X}\), write
\(
\mathcal F_{\mathcal R}(\mathcal H)
:=
\{R\cap\mathcal H : R\in\mathcal R\}
\)
for the induced set system on the population. We will apply
\Cref{lem:witness-uniform-convergence} to such induced set systems. The next
proposition specializes this to Euclidean balls, whose induced set system has
VC dimension at most \(d+1\); the same argument applies to any range family
\(\mathcal R\) with VC dimension \(h\), replacing \(d\) in the sample bound by
\(h\).
\begin{restatable}{proposition}{clusterprop}
\label{prop:cluster}Let \(\varepsilon\in(0,1/e]\) and \(\delta\in(0,1)\). Let \(\mathcal H\subseteq\mathbb R^d\) be a population of size \(n_H\), and let
\(\mathcal N\subseteq\mathcal H\) be a uniformly random \(n\)-element subset,
where
\(n=\Omega\!\left(
\varepsilon^{-2}
(d\log(1/\varepsilon)+\log(1/\delta))
\right)\).
Then, with probability at least \(1-\delta\), every mPJR+ slate
\(X^*\) with respect to \(\mathcal N\) satisfies the following: for every
\(\ell\in[k]\), every unselected candidate
\(c\in\mathcal N\setminus X^*\), and every ball-induced subset
\(S\subseteq\mathcal H\) with
\(|S|\ge(\ell/k+\varepsilon)n_H\), at least \(\ell\) members \(x\in X^*\)
satisfy
\[
\min_{i\in S}\Delta(i,x)
\le
\max_{i\in S}\Delta(i,c).
\]
\end{restatable}
\begin{proof}
Let \(\mathcal N\) be drawn uniformly without replacement from \(\mathcal H\).
Let \(\mathcal R\) be the range family of Euclidean balls in \(\mathbb R^d\),
and let \(\mathcal F:=\{R\cap\mathcal H:R\in\mathcal R\}\) be the induced
family of ball-induced subsets of the population. 
For \(S\in\mathcal F\), define
\(p(S):=\frac{|S|}{n_H}\) and
\(\hat p_{\mathcal N}(S):=\frac{|S\cap\mathcal N|}{n}\).

Since \(\mathcal F\) is induced by Euclidean balls and
\(\VCdim(\mathcal R)=d+1\), the same VC bound applies to \(\mathcal F\).
Thus, by the finite-population VC bound \Cref{lem:witness-uniform-convergence}, for
\[
n=\Omega\!\left(
\frac{1}{\varepsilon^2}
\left(d\log\frac{1}{\varepsilon}+\log\frac{1}{\delta}\right)
\right),
\]
we have, with probability at least \(1-\delta\),
\[
\sup_{S\in\mathcal F}|\hat p_{\mathcal N}(S)-p(S)|\le \varepsilon .
\]
On this event, if \(S\in\mathcal F\) satisfies \(p(S)\ge \ell/k+\varepsilon\),
then \(\hat p_{\mathcal N}(S)\ge \ell/k\). Fix an unselected candidate
\(c\in\mathcal N\setminus X^*\). Since \(X^*\) is an mPJR+ slate with
respect to \(\mathcal N\), applying mPJR+ to \(S\cap\mathcal N\) and \(c\)
gives at least \(\ell\) members \(x\in X^*\) such that
\[
\min_{i\in S\cap\mathcal N}\Delta(i,x)
\le
\max_{i\in S\cap\mathcal N}\Delta(i,c).
\]
Because \(S\cap\mathcal N\subseteq S\), we have
\(\min_{i\in S}\Delta(i,x)\le
\min_{i\in S\cap\mathcal N}\Delta(i,x)\) and
\(\max_{i\in S\cap\mathcal N}\Delta(i,c)\le
\max_{i\in S}\Delta(i,c)\) for each such \(x\). Hence
\[
\min_{i\in S}\Delta(i,x)
\le
\max_{i\in S}\Delta(i,c),
\]
proving the claim.
\end{proof}

\Cref{prop:cluster} gives population transfer for ball-induced groups.
We now extend this guarantee to arbitrary cohesive groups, at the cost of a
factor-\(3\) approximation in distance. This yields an alternative to
\Cref{thm:sampled-population} whose sample complexity depends on
\(d\), \(\varepsilon\), and \(\delta\), but is independent of \(k\) and \(n_H\).

\begin{theorem}
\label{thm:arbitrary-cohesive-transfer}Let \(\varepsilon\in(0,1/e]\) and \(\delta\in(0,1)\). Let \(\mathcal H\subset\mathbb R^d\) be a population of size
\(n_H\), and let \(k\) be the slate size. Let \(\varepsilon,\delta>0\), and let
\(\mathcal N\subseteq\mathcal H\) be a uniformly random \(n\)-element subset,
where
\[
n
=
\Omega\!\left(
\frac{1}{\varepsilon^2}
\left(
d\log\frac{1}{\varepsilon}
+
\log\frac{1}{\delta}
\right)
\right).
\]
Then, with probability at least \(1-\delta\), every mPJR+ slate
\(X^*\) with respect to \(\mathcal N\) satisfies the following: for every
\(\ell\in[k]\), every unselected candidate \(c\in\mathcal N\setminus X^*\), every
\(S\subseteq\mathcal H\), and every \(\tau>0\) such that
\(|S|\ge(\ell/k+\varepsilon)n_H\) and
\(\max_{i\in S}\Delta(i,c)\le \tau\),
\[
\left|
\left\{
x\in X^* \mid \exists i\in S,\ 
\Delta(i,x)\le 3\tau
\right\}
\right|
\ge \ell .
\]

\end{theorem}
\begin{proof}
By the assumed lower bound on \(|\mathcal N|\), \Cref{prop:cluster}
applies to the sampled population \(\mathcal N\). Hence, with probability at
least \(1-\delta\), every ball-induced subset \(S'\subseteq\mathcal H\) with
\(|S'|/n_H\ge \ell/k+\varepsilon\), and every unselected candidate
\(c\in\mathcal N\setminus X^*\), has at least \(\ell\) members \(x\in X^*\)
such that
\[
\min_{z\in S'}\Delta(z,x)
\le
\max_{z\in S'}\Delta(z,c).
\]
We condition on this event for the remainder of the proof.

Now consider a candidate \(c\), a group \(S\), and a value \(\tau\) as in the
statement of the theorem. Let \(B(c,\tau)\) be the Euclidean ball of radius \(\tau\)
centered at \(c\), and let \(S_B:=B(c,\tau)\cap\mathcal H\). Since
\(S\subseteq S_B\) and \(|S|\ge(\ell/k+\varepsilon)n_H\), we have
\(\frac{|S_B|}{n_H}\ge\frac{\ell}{k}+\varepsilon\). Moreover,
\(S_B\) is ball-induced. Therefore, by \Cref{prop:cluster}, there are at least
\(\ell\) members \(x\in X^*\) such that
\[
\min_{z\in S_B}\Delta(z,x)
\le
\max_{z\in S_B}\Delta(z,c)
\le \tau.
\]
For each such \(x\), choose \(z\in S_B\) with \(\Delta(z,x)\le \tau\), and fix
any \(i\in S\). Since both \(z\) and \(i\) lie in \(B(c,\tau)\), the triangle
inequality gives
\[
\Delta(i,x)
\le
\Delta(i,c)+\Delta(c,z)+\Delta(z,x)
\le 3\tau.
\]
Thus at least \(\ell\) members of \(X^*\) are within the claimed distance of
some member of \(S\).
\end{proof}

\paragraph{Stratified sampling.}
We next extend \Cref{thm:sampled-population} to stratified sampling, as used in
our experiments in \Cref{sec:exps}. The proof is a black-box conditioning
argument: any high-probability guarantee for a uniformly random sample can be
converted into the corresponding guarantee for a uniformly random stratified
sample, at the cost of a polynomial factor in the failure probability.

Suppose that, for some constant $t\geq 1$, we have a partition
$(P_1,\ldots,P_t)$ of $\mathcal H$ such that $|P_i|=r_i n_H$, where
$n_H:=|\mathcal H|$, $\sum_i r_i=1$, and
\(\rho:=\min_i r_i>0.
\)

\begin{definition}
We say that a subset $\mathcal N$ of $\mathcal H$ is stratified if
\[
|\mathcal N\cap P_i|
\in
\{\lceil r_i|\mathcal N|\rceil,\lfloor r_i|\mathcal N|\rfloor\}
\quad\text{for all }1\leq i\leq t.
\]
\end{definition}
Using \Cref{prop:reward-metric-equivalence} together with the monotone
relationship between angular and Euclidean distance on the unit sphere,
the uniform-sampling guarantee of \Cref{thm:sampled-population} implies
a \(4\)-approximate mPJR+ population-transfer guarantee in the induced
metric space. We now show that the same guarantee survives conditioning
on a stratified sample.
\begin{theorem}
Let $\varepsilon \in (0, 1/e]$ and $\delta \in (0,1)$. Let \(\mathcal N\) be sampled
uniformly at random without replacement among stratified
samples from \(\mathcal H\) of size \(n\), with
\[
n=\Omega\!\left(
\frac{1}{\varepsilon^2}
\left[
\log\frac{kn_H}{\delta}
+t\log\!\left(
\frac{t}{\varepsilon^2}\log\frac{kn_H}{\delta}
\right)
\right]
\right)\footnote{The hidden constant in the $\Omega(\cdot)$ notation may depend on $\rho$}.
\]
Then, with probability at least \(1-\delta\), every mPJR+ slate
\(X^*\) with respect to \(\mathcal N\) satisfies the following: for every
\(\ell\in[k]\), every unselected candidate
\(c\in\mathcal N\setminus X^*\), and every group
\(S\subseteq\mathcal H\) with
\(|S|\ge(\ell/k+\varepsilon)n_H\), at least \(\ell\) members
\(x\in X^*\) satisfy
\[
\min_{i\in S}\Delta(i,x)
\le 4\max_{i\in S}\Delta(i,c).
\]
\end{theorem}

\begin{proof}
We will prove the result by using \Cref{thm:sampled-population} in a black-box
fashion.

Let $U\subseteq\mathcal H$ be a uniformly random sample of size $n$. 
Let \(\mathsf B\) be the event that there exists an mPJR+ slate for \(U\)
that fails the stated \(4\)-approximate population-transfer guarantee.
 Let
$\mathsf A$ be the event that $U$ is stratified with respect to the partition
$(P_1,\ldots,P_t)$.
Set $n_{H,i}:=|P_i|$, so that $\sum_i n_{H,i}=n_H$
and $r_i=n_{H,i}/n_H$. For a draw of size $n$,
write feasible count vectors as
\[
\mathcal S
=
\Bigl\{
\mathbf c=(c_1,\dots,c_t)\in\mathbb Z_{\ge0}^t:
\sum_i c_i=n,\ 0\le c_i\le n_{H,i}
\Bigr\}.
\]

Let $\mathcal C\subseteq\mathcal S$ be the set of \emph{stratified} vectors:
$c_i\in\{\lfloor nr_i\rfloor,\lceil nr_i\rceil\}$ for all $i$.

The multivariate hypergeometric (sampling without replacement) probability of a count vector
\(\mathbf c\in\mathcal S\) is

\[
H(\mathbf c)
=
\frac{\prod_{i=1}^t\binom{n_{H,i}}{c_i}}
{\binom{n_H}{n}}.
\]
We claim that
\[
\Pr(\mathsf A)
\ge H(\mathbf c^\star)
\ge \Omega\!\left(n^{-t}\right)
\]
for some stratified vector $\mathbf c^\star$, provided $n$ is large enough.

\medskip

\begin{claim}\label{clm:hyper-lb-explicit}
Fix $t\ge2$ and $\rho\in(0,1)$ and assume
$r_i:=n_{H,i}/n_H\ge\rho$ for all $i$.
Let $\mathbf c=(c_1,\dots,c_t)$ satisfy\(
c_i\in\{\lfloor nr_i\rfloor,\lceil nr_i\rceil\}
\quad\text{for all }i,\)
where $n=\sum_i c_i\in\{1,\dots,n_H-1\}$.
Then, for all $n$ with $\min\{n,n_H-n\}\ge2/\rho$,
\[
H(\mathbf c)
\ge
Z'_t\,n^{-\frac{t-1}{2}},
\qquad
Z'_t
:=
(2\pi)^{\frac{1-t}{2}}
\rho^{t/2}
\exp\!\left(-\frac{t}{6}-\frac1{12}\right).
\]
For $\min\{n,n_H-n\}<2/\rho$, there is a constant
$Z''_{t,\rho}>0$, depending only on $t$ and $\rho$,
such that $H(\mathbf c)\ge Z''_{t,\rho}$.
\end{claim}

\begin{proof}
Recall
\[
H(\mathbf c)
=
\frac{\prod_{i=1}^t\binom{n_{H,i}}{c_i}}{\binom{n_H}{n}}
=
\frac{\prod_{i=1}^t\frac{n_{H,i}!}{c_i!(n_{H,i}-c_i)!}}
{\frac{n_H!}{n!(n_H-n)!}}
=
\frac{\bigl(\prod_i n_{H,i}!\bigr)n!(n_H-n)!}
{n_H!\bigl(\prod_i c_i!\bigr)
\bigl(\prod_i(n_{H,i}-c_i)!\bigr)}.
\tag{$\dagger$}
\]
We use Robbins' version of Stirling's bounds~\citep{robbins1955stirling}:
for every integer $q\ge1$,
\[
\sqrt{2\pi q}\left(\frac qe\right)^q
\le
q!
\le
\sqrt{2\pi q}\left(\frac qe\right)^q e^{1/(12q)}.
\]
Apply the \emph{lower} bound to $n_{H,i}!,n!,(n_H-n)!$
and the \emph{upper} bound to $n_H!,c_i!,(n_{H,i}-c_i)!$
in~$(\dagger)$. Plugging in and cancelling the
common $(q/e)^q$-terms, we obtain
\[
H(\mathbf c)
\ge
(2\pi)^{\frac{1-t}{2}}
\frac{\sqrt n\sqrt{n_H-n}\prod_i\sqrt{n_{H,i}}}
{\sqrt{n_H}\prod_i\sqrt{c_i}\prod_i\sqrt{n_{H,i}-c_i}}
\exp(-E_t),
\tag{$\ast$}
\]
where the error exponent is
\[
E_t
:=
\frac1{12n_H}
+
\sum_{i=1}^t
\left(
\frac1{12c_i}
+
\frac1{12(n_{H,i}-c_i)}
\right).
\]

We first lower-bound
\[
R
:=
\frac{\sqrt n\sqrt{n_H-n}\prod_i\sqrt{n_{H,i}}}
{\sqrt{n_H}\prod_i\sqrt{c_i}\prod_i\sqrt{n_{H,i}-c_i}}.
\]
Use the assumptions $r_i=n_{H,i}/n_H\ge\rho$,
$\sum_i n_{H,i}=n_H$, and $\sum_i c_i=n$.
For the numerator,
\[
\prod_i\sqrt{n_{H,i}}
=
\sqrt{\prod_i n_{H,i}}
\ge
\sqrt{(n_H\rho)^t}
=
n_H^{t/2}\rho^{t/2}.
\]
For the denominators, we use only crude upper bounds.
Because $c_i\le n$ and $\sum_i c_i=n$, we have
$\prod_i\sqrt{c_i}\le n^{t/2}$.
Because $0\le n_{H,i}-c_i$ and
$\sum_i(n_{H,i}-c_i)=n_H-n$, each
$n_{H,i}-c_i\le n_H-n$, so
$\prod_i\sqrt{n_{H,i}-c_i}\le(n_H-n)^{t/2}$.
Thus
\[
R
\ge
\frac{\sqrt n\sqrt{n_H-n}\,n_H^{t/2}\rho^{t/2}}
{\sqrt{n_H}\,n^{t/2}(n_H-n)^{t/2}}
=
\rho^{t/2}
\frac{n_H^{\frac{t-1}{2}}}
{n^{\frac{t-1}{2}}(n_H-n)^{\frac{t-1}{2}}}.
\]
Since $n_H-n\le n_H$, we have
$n_H^{(t-1)/2}/(n_H-n)^{(t-1)/2}\ge1$,
and hence $R\ge\rho^{t/2}n^{-(t-1)/2}$.

We now impose $\min\{n,n_H-n\}\ge2/\rho$ as in the statement.
For a stratified vector
$c_i\in\{\lfloor nr_i\rfloor,\lceil nr_i\rceil\}$
with $r_i\ge\rho$, we have
\[
c_i\ge nr_i-1\ge n\rho-1\ge\tfrac12n\rho,
\]
because $n\rho\ge2$. Similarly,
\[
n_{H,i}-c_i
=
n_H r_i-c_i
\ge
n_H r_i-(nr_i+1)
=
(n_H-n)r_i-1
\ge
\tfrac12(n_H-n)\rho.
\]
Therefore,
\[
\frac1{12c_i}\le\frac1{6n\rho},
\qquad
\frac1{12(n_{H,i}-c_i)}
\le\frac1{6(n_H-n)\rho},
\]
and hence
\[
\sum_{i=1}^t
\left(
\frac1{12c_i}
+
\frac1{12(n_{H,i}-c_i)}
\right)
\le
\frac{t}{6n\rho}
+
\frac{t}{6(n_H-n)\rho}
\le
\frac{t}{3\min\{n,n_H-n\}\rho}
\le
\frac t6.
\]
Finally, since $n_H\ge n\ge2/\rho$, we also have
$1/(12n_H)\le1/12$.
Altogether, $E_t\le1/12+t/6$, and therefore
\[
\exp(-E_t)
\ge
\exp\!\left(-\frac t6-\frac1{12}\right).
\]

Substituting these bounds into~$(\ast)$ yields
\[
H(\mathbf c)
\ge
(2\pi)^{\frac{1-t}{2}}
\rho^{t/2}
\exp\!\left(-\frac t6-\frac1{12}\right)
n^{-\frac{t-1}{2}},
\]
which is the claimed bound.
It remains to consider $\min\{n,n_H-n\}<2/\rho$.
Suppose first that $n<2/\rho$.
Since $c_i\le n$ and $n_{H,i}\ge\rho n_H$, the standard bounds
$\binom{a}{b}\ge(a/b)^b$ and $\binom{n_H}{n}\le n_H^n$ give
\[
H(\mathbf c)
\ge
\frac{\prod_{i:c_i>0}(n_{H,i}/c_i)^{c_i}}{n_H^n}
\ge
\left(\frac{\rho}{n}\right)^n,
\]
which is bounded below by a positive constant depending only on $\rho$.
If instead $n_H-n<2/\rho$, apply the same argument to the complement count
vector $(n_{H,i}-c_i)_{i=1}^t$, which is stratified for a draw of size
$n_H-n$ and has the same hypergeometric probability.
This gives the claimed constant lower bound in the remaining regime.

\end{proof}

Let $Z_t:=\min\{Z'_t,Z''_{t,\rho}\}$. We thus get the direct lower bound
\[
\Pr_{\mathrm{WOR}}(\mathsf A)
=
\sum_{\mathbf c\in\mathcal C}H(\mathbf c)
\ge
H(\mathbf c^\star)
\ge
Z_t n^{-\frac{t-1}{2}}.
\tag{$\dagger$}
\]

\medskip
Let $\delta'>0$ be a target failure probability for
\Cref{thm:sampled-population} under sampling without replacement. Applying
\Cref{thm:sampled-population} with failure probability $\delta'$, if
\[
n=\Omega\!\left(
\frac1{\varepsilon^2}\log\frac{kn_H}{\delta'}
\right),
\]
then a uniformly random sample $U$ of size $n$ is bad with probability at most
$\Pr_{\mathrm{WOR}}(\mathsf B)\le\delta'$.

Set $\delta'=\delta\Pr(\mathsf A)$ and use~$(\dagger)$ to bound
$\Pr(\mathsf A)$ from below. This yields
\[
n=\Omega\!\left(
\frac1{\varepsilon^2}
\left(
\log\frac{kn_H}{\delta}+t\log n
\right)
\right).
\]
Equivalently, by a standard bootstrapping argument, this is implied by
\[
n=\Omega\!\left(
\frac1{\varepsilon^2}
\left[
\log\frac{kn_H}{\delta}
+
t\log\!\left(
\frac{t}{\varepsilon^2}\log\frac{kn_H}{\delta}
\right)
\right]
\right).
\]
With this choice, $\Pr(\mathsf B)\le\delta'$. Therefore,
\[
\Pr(\mathsf B\mid\mathsf A)
=
\frac{\Pr(\mathsf B\cap\mathsf A)}{\Pr(\mathsf A)}
\le
\frac{\Pr(\mathsf B)}{\Pr(\mathsf A)}
\le
\frac{\delta'}{\Pr(\mathsf A)}
=
\delta.
\]
Thus, conditional on the sample being stratified
(equivalently, sampling uniformly without replacement among stratified
samples), every mPJR+ slate for the sample satisfies the stated $4$-approximate
mPJR+ guarantee for the full population with probability at least $1-\delta$.
\end{proof}

   \section{Experimental Details}

\subsection{Remesh Dataset Details} \label{app:remesh-dataset}

\subsubsection{Topic Generation Prompt}
\label{app:topic-gen-prompt}

The following prompt was given to Claude Opus 4.6\footnote{Claude Opus 4.6~\citep{anthropic2024claude} is released under a commercial license by Anthropic.} to generate ten diverse topics for evaluation:

\begin{promptbox}
Here are four questions that a group of participants was asked about:

What are your impressions of the campus protests?

What are your impressions of how campus administrators handled the protests?

What are your thoughts on the way university campus administrators should approach the issue of Israel/Gaza demonstrations?

What should guide university campus administrators handling of protests?

Given a person's answer to these questions, what are ten other diverse topics would be predictable*?

Return only a JSONL, where each row is formatted as \{``topic": topic\}. The topics you generated will be given to an LLM to generate a variety of questions on that topic, so give enough information about what each topic represents in the JSONL.
\end{promptbox}

\subsubsection{Prompt Generation Prompt}
\label{app:prompt-gen-prompt}

For each of the ten generated topics, the following prompt template was given to Claude Opus 4.6 to generate 100 diverse opinion questions:

\begin{promptbox}
Generate 100 diverse and open-ended questions asking people for their opinions on the following topic: \{topic\}. Choose a mix of contentious questions that people are likely divided on and non-contentious questions that most people would be in agreement on. You should make each question stand-alone, so that a person can understand what you are referring to without additional context. Your output format should be a JSONL, where each line is formatted as \{``prompt\_id": prompt\_id, ``prompt": prompt\}
\end{promptbox}

\subsubsection{Generated Topics}
\label{app:generated-topics}

The following ten topics were generated by Claude Opus 4.6:

\begin{enumerate}
    \item Views on free speech and the First Amendment, including where the line should be drawn between protected expression and harassment or disruption.
    \item Opinions on the Israel--Palestine conflict, including views on Israeli military operations in Gaza and Palestinian statehood.
    \item Attitudes toward higher education institutions, including trust in university leadership, the purpose of universities, and whether colleges are too liberal or too conservative.
    \item Views on student activism and civil disobedience throughout history, including whether disruptive protest is an effective or justified means of social change.
    \item Opinions on police use of force in crowd control and protest situations, including whether law enforcement intervention on campuses is appropriate.
    \item Attitudes toward U.S. foreign policy and military aid, particularly regarding funding and arms sales to allied nations in conflict zones.
    \item Views on antisemitism and Islamophobia, including how to define hate speech, whether certain political criticisms cross into bigotry, and how institutions should address discrimination.
    \item Opinions on generational differences in political engagement, including whether younger generations are more or less informed, idealistic, or effective in their activism compared to older generations.
    \item Attitudes toward media coverage and bias, including whether mainstream and social media fairly represent protest movements, campus issues, and the Israel--Gaza conflict.
    \item Views on corporate and institutional divestment as a political tool, including whether universities, pension funds, and businesses should divest from industries or nations involved in controversial practices.
\end{enumerate}

\subsubsection{Simulated Agent Prompt}
\label{app:agent-prompt}

Each simulated agent is implemented via a roleplay prompt to Llama-3.1-8B-Instruct~\citep{grattafiori2024llama}\footnote{Llama-3.1-8B-Instruct~\citep{grattafiori2024llama} is released under a commercial license from Meta.} with the following system and user prompts.

\paragraph{System prompt.}
\begin{promptbox}
You are roleplaying as the following person:

\textit{\{agent description\}}

Generate a genuine response that reflects this person's personal perspective. Do not hedge or speak in generalities --- respond as this specific person would.
\end{promptbox}

The participant description is constructed from the participant's demographic attributes in the Remesh survey (e.g., gender, age, political affiliation, religious affiliation, education level, urbanicity, household income).

\paragraph{User prompt.}
The user prompt is composed of three optional sections, joined by blank lines:

\begin{enumerate}
\item \textbf{Previous responses} (if available):
\begin{promptbox}
Here are questions you answered previously:

Question: \textit{\{question text\}} \\
Response: \textit{\{participant's response\}}

\textit{[...repeated for each prior question...]}
\end{promptbox}

\item \textbf{Voting context} (if available):
\begin{promptbox}
Here are other people's statements that you agreed with:

Question: \textit{\{question text\}} \\
Statement: \textit{\{statement text\}}

\textit{[...repeated for each agreed statement...]}

Here are other people's statements that you disagreed with:

Question: \textit{\{question text\}} \\
Statement: \textit{\{statement text\}}

\textit{[...repeated for each disagreed statement...]}
\end{promptbox}

\item \textbf{Target question}:
\begin{promptbox}
Now answer this new question in 2-3 sentences.

Question: \textit{\{target question text\}} \\
Response:
\end{promptbox}
\end{enumerate}

\subsubsection{Viewpoint Summary Prompt}
\label{app:viewpoint-summary-prompt}

For each participant, we summarize their simulated agent's responses on the training prompts into a concise viewpoint summary using Claude Opus 4.6. Responses are grouped by topic, and each topic is summarized separately using the following prompt:

\paragraph{System prompt.}
\begin{promptbox}
You are a helpful assistant that summarizes a person's viewpoints. Given a set of questions and the person's responses on a particular topic, write a concise summary ($<=3$ sentences) capturing their key views, attitudes, and reasoning on this topic. Do not list individual responses --- synthesize them into a coherent summary.
\end{promptbox}

\paragraph{User prompt.}
\begin{promptbox}
Topic: \textit{\{topic\}}

Here are the person's responses:

Q: \textit{\{question text\}} \\
A: \textit{\{viewpoint text\}}

\textit{[...repeated for each response on this topic...]}

Summarize this person's views on this topic in three sentences or less.
\end{promptbox}

The per-topic summaries are then concatenated into a single viewpoint summary per participant, formatted as \texttt{**<topic>**: <summary>} for each topic. An example viewpoint summary for one simulated participant is shown below.

\paragraph{Example viewpoint summary.}
\begin{promptbox}
\textbf{Views on free speech and the First Amendment:} This person supports broad free speech protections in principle but draws the line at speech that incites violence, constitutes harassment, or promotes hateful extremist ideologies like white supremacy—favoring practical measures like buffer zones, anti-mask laws, and free speech zones to balance expression with safety and order. They are skeptical of cancel culture, speech codes, and NDAs that silence misconduct reporting, viewing these as threats to open discourse, while also opposing corporate free speech rights and believing the digital age requires updated legal frameworks. Their views reflect a conservative Christian perspective that values individual liberty and democratic institutions but accepts significant exceptions when speech causes tangible harm, disrupts communities, or undermines democratic values.

\textbf{Opinions on the Israel-Palestine conflict:} This person admits to having limited knowledge about the Israel-Palestine conflict and approaches it primarily through their Christian faith, emphasizing prayer, peace, and treating all people with dignity. They hold generally balanced views, believing both sides share responsibility though leaning slightly toward seeing Hamas as more aggressive, while expressing sympathy for Palestinian civilians and opposing the destruction of civilian infrastructure. They favor a two-state solution and peaceful dialogue, are skeptical of Netanyahu's commitment to peace, oppose embassy moves to Jerusalem and the BDS movement, and are cautiously optimistic that future generations may achieve peace.

\textbf{Attitudes toward higher education institutions:} This person is skeptical of university leadership, viewing administrators as more focused on money, image, and progressive politics than on delivering quality education, and believes campuses have a liberal bias with DEI initiatives and speech codes that have gone too far. They value practical, affordable education—supporting community colleges and trade alternatives—while questioning whether four-year degrees are worth the cost, and they strongly favor free speech, professor neutrality, and accountability over tenure protections and coddling measures like safe spaces. Overall, they hold a conservative-leaning, populist perspective that universities should return to their core educational mission, be more accessible to working-class and rural students, and stop prioritizing institutional interests and ideological agendas over student learning.

\textbf{Views on student activism and civil disobedience throughout history:} This person generally favors working within the system and peaceful, non-disruptive protest, expressing frustration when activism inconveniences others or involves property destruction, while believing protesters should face legal consequences for lawbreaking. However, they hold somewhat contradictory views—admiring historical civil disobedience movements like the civil rights sit-ins and freedom riders as courageous and justified, while simultaneously maintaining that law-breaking is never morally justified and that students today should use less disruptive methods. Overall, they value order and respect for authority but acknowledge that history often vindicates protesters, and they draw a distinction between morally-motivated activism targeting genuine injustice versus ideologically-driven disruption, while expressing skepticism about young people's experience and the effectiveness of confrontational tactics.

\textbf{Opinions on police use of force in crowd control and protest situations:} This person supports a graduated, de-escalation-first approach to police use of force at protests, believing force should only be used when situations turn violent or destructive, and strongly opposes heavy-handed tactics like tear gas, pepper spray on non-violent protesters, militarization, and surveillance technologies like drones and facial recognition. They believe universities can call in outside law enforcement as a last resort but oppose external pressure (like donors) driving that decision, support students' right to protest without academic penalty, and think clear policies, transparency, and independent oversight are essential for accountability. While generally moderate and pragmatic—accepting protest permits, predetermined routes, and confiscation of potential weapons as reasonable—they consistently prioritize free speech protections, fairness regardless of political ideology, and holding individual officers personally liable for excessive force.

\textbf{Attitudes toward U.S. foreign policy and military aid:} This person favors a cautious, values-driven approach to U.S. military aid, believing it should prioritize democratic allies who share American values while requiring greater transparency, congressional oversight, and public accountability. Rooted in Christian faith and conservative principles, they are skeptical that military aid promotes stability—viewing it as often escalating conflicts and primarily benefiting defense industry elites rather than ordinary Americans—and strongly prefer diplomatic solutions over military ones. While they support maintaining key alliances (Israel, Taiwan, Ukraine) and military deterrence through bases and defense systems, they want stricter conditions on aid, including caps on spending, humanitarian considerations, and recipient accountability.

\textbf{Views on antisemitism and Islamophobia:} This person takes a moderate, faith-informed approach to antisemitism and Islamophobia, believing both are serious forms of hatred that deserve equal attention and should be addressed through education, dialogue, and mutual respect rather than censorship. They consistently distinguish between legitimate political criticism (of Israel or Islamic practices) and genuine bigotry, while acknowledging the "fine lines" and "gray areas" involved, and they favor contextualization over censorship in most cases. They express some skepticism about the effectiveness of institutional solutions—whether interfaith programs, corporate statements, or AI moderation—preferring personal relationships and individual accountability as paths toward reducing prejudice, while supporting enhanced hate crime penalties and platform responsibility for removing clear hate speech.

\textbf{Opinions on generational differences in political engagement:} This person generally believes younger generations are too idealistic, impatient, and lacking in real-world experience to be as politically effective as older generations, whose consistent voting, strategic patience, and willingness to compromise they view as more impactful. They see young people as overly influenced by social media and peer pressure rather than critical thinking, and expect most will naturally become more conservative with age and life experience. However, they maintain a degree of fairness—acknowledging young people's passion deserves respect, defending equal voting rights regardless of age, and recognizing that older generations can sometimes be dismissive or stuck in the past.

\textbf{Attitudes toward media coverage and bias:} This person believes mainstream and social media both fail to provide balanced coverage, with mainstream outlets sensationalizing protests and prioritizing spectacle over substance, while social media amplifies misinformation and creates echo chambers that leave people less informed. They view media bias as largely inevitable—driven by corporate ownership, political funding, and the 24-hour news cycle—but strongly advocate for transparency, clear labeling of opinion versus news, and disclosure of affiliations and biases. They generally trust traditional reputable news sources over social media or independent media, favor hiring regional journalists and teaching media literacy, and believe coverage of the Israel-Gaza conflict focuses too heavily on politics and extreme voices while neglecting humanitarian impacts and moderate perspectives.

\textbf{Views on corporate and institutional divestment as a political tool:} This person is generally skeptical of divestment as a political tool, viewing it as largely symbolic and ineffective, and believes institutions like universities and pension funds should prioritize their fiduciary responsibilities—maximizing returns for students and retirees—over making political statements. However, they hold nuanced exceptions: they support divestment from industries that directly contradict an institution's mission (like hospitals divesting from tobacco or municipalities from private prisons), from authoritarian regimes on moral grounds, and they value transparency in investment decisions. Their overall stance reflects a pragmatic, conservative-leaning perspective shaped by Christian values and rural life, emphasizing that grassroots activism and engagement are more effective than financial gestures, while expressing concern about slippery slopes, unintended harm to ordinary people, and the elitist nature of divestment campaigns.
\end{promptbox}

\subsection{Proxy Embedding Training Details} \label{app:training-details}
For the viewpoint embedding model $\psi$, we use the \texttt{embeddings-for-preferences-st5-xl} created by \citet{prefembeddings2026}, which is an embedding model fine-tuned so that distances between embeddings are reflective of preference rather than semantic similarity. To train the proxy embedding model $\theta$,  we fine-tune a copy of \texttt{embeddings-for-preferences-st5-xl} to take the user context and prompt as input and predict that user's ground-truth response embedding for that prompt.

\paragraph{Input format.}
For each participant--prompt pair, the encoder input is:
\begin{promptbox}
\small
\texttt{User:} \textit{\{user context\}} \\[4pt]
\texttt{Question:} \textit{\{prompt\}}
\end{promptbox}
For Remesh, the participant description is a summary of the participant's viewpoints on the training prompts (see \Cref{app:viewpoint-summary-prompt}); for Reddit AITA, it is the summary of the user's Reddit history from \citet{wu2026humanlm}.

\paragraph{Loss function.}
The training loss combines (i) elementwise mean-squared error between predicted and target embeddings and (ii) a Pearson correlation loss that encourages the predictor to preserve the pairwise distance structure between embeddings. For a prompt $p$, let $\hat{\mathbf{e}}_i = \theta(u_i,p)\in\mathbb{R}^d$ be the predicted embedding for participant $i$, where $u_i$ is the user context, and let $\mathbf{e}_i$ be the corresponding ground-truth response embedding for that user on that prompt. We sample a batch of $s$ users for the prompt and optimize the loss:
\begin{align}
    \mathcal{L}(\hat{\mathbf{E}},\mathbf{E}) ={}&
    \underbrace{\frac{1}{sd}\sum_{i=1}^{s}\|\hat{\mathbf{e}}_i-\mathbf{e}_i\|_2^2}_{\mathcal{L}_{\mathrm{MSE}}} \\
    &+\lambda\underbrace{\left[1-\operatorname{Pearson}\!\left(\mathbf{d}(\hat{\mathbf{E}}),\mathbf{d}(\mathbf{E})\right)\right]}_{\mathcal{L}_{\mathrm{Pearson}}},
\end{align}
where $\hat{\mathbf{E}}=(\hat{\mathbf{e}}_1,\ldots,\hat{\mathbf{e}}_s)$, $\mathbf{E}=(\mathbf{e}_1,\ldots,\mathbf{e}_s)$, and $\mathbf{d}(\mathbf{E})\in\mathbb{R}^{\binom{s}{2}}$ contains all pairwise Euclidean distances. We train with $\lambda=1$.

\paragraph{Optimization.}
We fine-tune the embedding model using LoRA~\citep{hu2022lora} adapters on the query and value attention modules (rank $r=8$, $\alpha=16$, and dropout $0.05$). We use AdamW~\citep{loshchilov2018decoupled} with learning rate $10^{-5}$, weight decay $0.01$, no gradient accumulation, and gradient clipping at norm $1.0$. We train for one epoch with seed 42 and a maximum sequence length of 2,048 tokens. Each training batch consists of 16 user responses to the same prompt. For both datasets, we train using distributed data parallelism over eight H100 GPUs and bfloat16 autocasting, with gradient checkpointing enabled.

\subsection{Qualitative Examples}
\label{app:qual-examples}

Here, we present three qualitative examples for each dataset.

\begin{enumerate}
    \item \textbf{Remesh:} ``Should pension fund managers be legally allowed to consider ethical or political factors when making investment decisions, even if it might slightly reduce returns?'' See \Cref{fig:remesh-rank04-pca} and \Cref{tab:remesh-rank04-pca-slates}.
    \item \textbf{Remesh:} ``Should students be allowed to form encampments on university property as a form of extended protest?'' See \Cref{fig:remesh-rank06-pca} and \Cref{tab:remesh-rank06-pca-slates}.
    \item \textbf{Remesh:} ``Should high school students be encouraged or discouraged from participating in political activism and civil disobedience?'' See \Cref{fig:remesh-rank01-pca} and \Cref{tab:remesh-rank01-pca-slates}.
    \item \textbf{Reddit:} ``AITA for my reaction when I learned that my fiance returned my wedding dress and replaced it with the one his mom picked for me?'' See \Cref{fig:reddit-rank01-pca} and \Cref{tab:reddit-rank01-pca-slates}.
    \item \textbf{Reddit:} ``AITA for making a dad joke?'' See \Cref{fig:reddit-rank03-pca} and \Cref{tab:reddit-rank03-pca-slates}.
    \item \textbf{Reddit:} ``AITA for telling a friend's boyfriend her intentions of getting pregnant?'' See \Cref{fig:reddit-rank06-pca} and \Cref{tab:reddit-rank06-pca-slates}.
\end{enumerate}

To obtain the most informative 2D projections, we focused on prompts for which the first two principal components (PC1 and PC2) explain comparatively high variance—specifically, those in the top ten within each dataset. Even among these, the first two PCs explain up to 76\% of the variance for Remesh prompts, but only up to 38\% for Reddit. This could be because both the Reddit dataset (\texttt{humanual-opinion}) and the HumanLM simulator trained on it are strongly skewed toward ``Not the Asshole" (NTA) judgments (see \Cref{fig:reddit-confusion-matrix} below). Lack of diversity in judgments makes the differences between viewpoints less meaningful, whereas, on Remesh, the first PC dimension typically corresponds to agreement vs disagreement with the premise of the question (e.g., see the example explained in \Cref{sec:exps}). Overall, these characteristics generally make the Remesh examples easier to interpret than the Reddit ones.

\begin{figure}[h]
    \centering
    \includegraphics[width=0.45\linewidth]{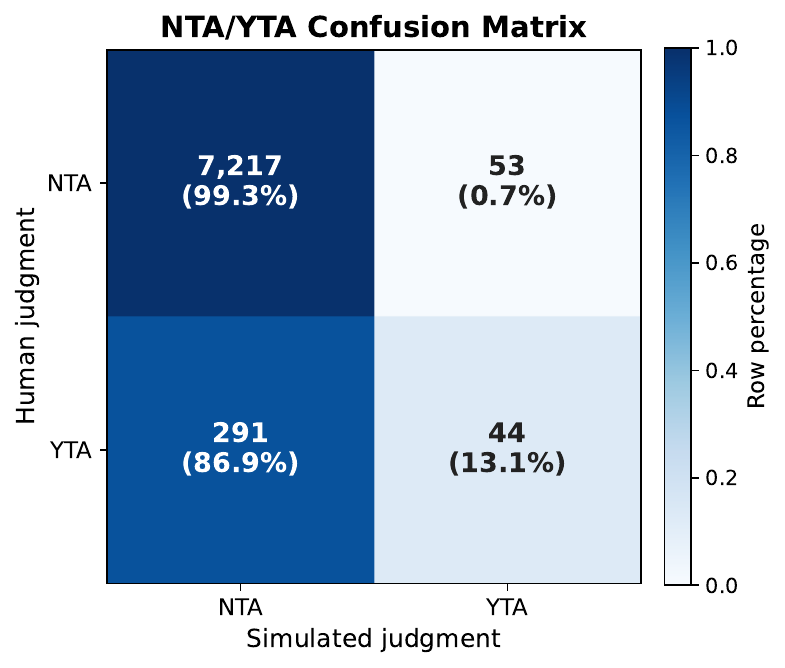}
    \caption{\textbf{Confusion matrix of NTA/YTA judgments for Reddit.} On the $y$-axis is the ground-truth human judgment from the Reddit dataset (\texttt{humanual-opinion}) and on the $x$-axis is HumanLM's simulated judgment. The dataset has a heavy skew towards human judgments of ``Not the Asshole'' (NTA), but for the judgments where the ground-truth human judgment is ``You're the Asshole'' (YTA), the HumanLM simulator also gets it wrong and predicts NTA $87\%$ of the time. This confusion matrix includes only response pairs where both the human and simulator outputs contained either ``NTA'' or ``YTA.''}
    \label{fig:reddit-confusion-matrix}
\end{figure}

\subsubsection{Remesh Qualitative Examples}

{\footnotesize
\begin{longtable}{p{0.22\textwidth} p{0.71\textwidth}}
    \caption{Selected responses for the Remesh prompt ``Should pension fund managers be legally allowed to consider ethical or political factors when making investment decisions, even if it might slightly reduce returns?'' ($k=4$). Embeddings are visualized in \Cref{fig:remesh-rank04-pca}.}
    \label{tab:remesh-rank04-pca-slates} \\
    \toprule
    \textbf{Method} & \textbf{Selected responses} \\
    \midrule
    \endfirsthead
    \toprule
    \textbf{Method} & \textbf{Selected responses (continued)} \\
    \midrule
    \endhead
    \begin{minipage}[t]{\linewidth}
    \textbf{SEAR (ours)}\par
    {\scriptsize\textsf{Route-then-Simulate}}\par\smallskip
    {\scriptsize\color{boxgreendark}\bfseries Satisfies mPJR+}
    \end{minipage} &
    \begin{enumerate}[leftmargin=*,nosep,label=\arabic*.]
        \item I think pension fund managers should be allowed to consider ethical or political factors when making investment decisions, as long as it doesn't compromise the financial stability of the fund. It's their responsibility to make decisions that align with the values of the people they're serving, and considering these factors can actually lead to more sustainable and long-term returns.
        \item I think pension fund managers should be allowed to consider ethical or political factors when making investment decisions, as long as it doesn't compromise the overall financial health of the fund. In fact, I believe it's their responsibility to consider the broader social and environmental implications of their investments, even if it means slightly reduced returns. This approach can help ensure that investments align with the values and goals of the fund's beneficiaries.
        \item I think pension fund managers should be allowed to consider ethical or political factors when making investment decisions, as long as it doesn't compromise the financial stability of the fund. It's a delicate balance, but I believe it's essential to consider the broader social implications of our investments, especially when it comes to issues like human rights and environmental sustainability. Ultimately, I think it's a matter of being a responsible steward of the funds, rather than just prioritizing short-term gains.
        \item I think pension fund managers should be allowed to consider ethical or political factors when making investment decisions, as long as it doesn't compromise the financial stability of the fund. It's their responsibility to make decisions that align with the values of the people they're serving, and if that means sacrificing a small amount of return for the sake of social responsibility, then so be it. Ultimately, it's about finding a balance between financial prudence and moral integrity.
    \end{enumerate} \\
    \midrule
    \begin{minipage}[t]{\linewidth}
    \textbf{$k$-means}\par
    {\scriptsize\textsf{Route-then-Simulate}}\par\smallskip
    {\scriptsize\color{boxreddark}\bfseries Does not satisfy mPJR+}
    \end{minipage} &
    \begin{enumerate}[leftmargin=*,nosep,label=\arabic*.]
        \item I think pension fund managers should be allowed to consider ethical or political factors when making investment decisions, as long as it doesn't compromise the financial stability of the fund. It's a matter of personal values and conscience, and I believe that's what being an independent thinker is all about. If it means sacrificing a small amount of return for the sake of doing what's right, then so be it.
        \item I think pension fund managers should be allowed to consider ethical or political factors when making investment decisions. It's their responsibility to ensure that the funds they manage are used in a way that aligns with the values of the people they serve, even if it means sacrificing a small amount of return. In fact, I think it's their duty to do so, and it's not necessarily about the returns, but about making a positive impact.
        \item No, I don't think pension fund managers should be allowed to consider ethical or political factors when making investment decisions. Their primary responsibility is to maximize returns for the fund's beneficiaries, and allowing them to make decisions based on personal opinions could lead to biased and potentially costly choices. It's best to stick to objective, data-driven decisions that prioritize financial performance.
        \item As a Christian and a conservative, I believe that pension fund managers should be allowed to consider ethical or political factors when making investment decisions, as long as it aligns with their fiduciary duty to their beneficiaries. In fact, I think it's their moral obligation to do so, especially when it comes to issues like human rights and social justice. By considering these factors, they can help ensure that their investments are not only financially sound but also morally responsible.
    \end{enumerate} \\
    \midrule
    \begin{minipage}[t]{\linewidth}
    \textbf{Random}\par
    {\scriptsize\textsf{Route-then-Simulate}}\par\smallskip
    {\scriptsize\color{boxreddark}\bfseries Does not satisfy mPJR+}
    \end{minipage} &
    \begin{enumerate}[leftmargin=*,nosep,label=\arabic*.]
        \item Honestly, I don't think it's right for pension fund managers to consider their own personal views when making investment decisions. They should be focused on getting the best returns for the people who rely on those pensions, not on pushing their own agendas. It's not their place to impose their own politics on the investments.
        \item I believe pension fund managers should be allowed to consider ethical or political factors when making investment decisions, as it reflects the values and principles of the individuals who contribute to the fund. By taking a socially responsible approach, they can align their investments with the moral compass of their stakeholders, even if it means slightly reduced returns. This approach can also help to promote long-term sustainability and stability in the market.
        \item I think pension fund managers should be allowed to consider ethical or political factors when making investment decisions, as long as it doesn't compromise the financial stability of the fund. It's a delicate balance between making a profit and doing what's right, but I believe it's possible to achieve both. Ultimately, the goal of a pension fund is to provide for people's retirement, not just to maximize returns.
        \item No, I don't think pension fund managers should be allowed to consider ethical or political factors when making investment decisions. Their primary responsibility is to ensure the financial security of the pensioners, and allowing them to make decisions based on personal opinions could compromise the fund's returns and stability. It's best to stick to a purely financial approach to ensure the long-term sustainability of the pension fund.
    \end{enumerate} \\
    \midrule
    \begin{minipage}[t]{\linewidth}
    \textbf{SEAR (oracle)}\par
    {\scriptsize\textsf{Simulate-then-Route}}\par\smallskip
    {\scriptsize\color{boxgreendark}\bfseries Satisfies mPJR+}
    \end{minipage} &
    \begin{enumerate}[leftmargin=*,nosep,label=\arabic*.]
        \item I think pension fund managers should be allowed to consider ethical or political factors when making investment decisions, as long as it doesn't compromise the overall financial stability of the fund. It's a delicate balance between financial returns and social responsibility, but I believe it's worth exploring options that align with the values of the fund's stakeholders. Ultimately, it's about making informed decisions that benefit both the investors and the broader community.
        \item I think pension fund managers should be allowed to consider ethical or political factors when making investment decisions. It's their responsibility to ensure that the funds they manage are used in a way that aligns with the values of the people they serve, and considering these factors can actually lead to more sustainable and long-term returns. By prioritizing social and environmental responsibility, they can help create a more equitable and just society, which ultimately benefits everyone.
        \item I believe pension fund managers should be allowed to consider ethical or political factors when making investment decisions, as long as it doesn't compromise the financial stability of the fund. In fact, I think it's their responsibility to consider the broader social and environmental implications of their investments, even if it means slightly reduced returns. This approach can help ensure that the fund's investments align with the values of its beneficiaries and contribute to a more sustainable future.
        \item I think pension fund managers should be allowed to consider ethical or political factors when making investment decisions, as long as it doesn't compromise the overall financial stability of the fund. It's a matter of personal values and conscience, and I believe that's what being a responsible steward of someone's retirement savings is all about.
    \end{enumerate} \\
    \midrule
    \bottomrule
\end{longtable}
}

\begin{figure*}[h]
    \centering
    \includegraphics[width=\textwidth]{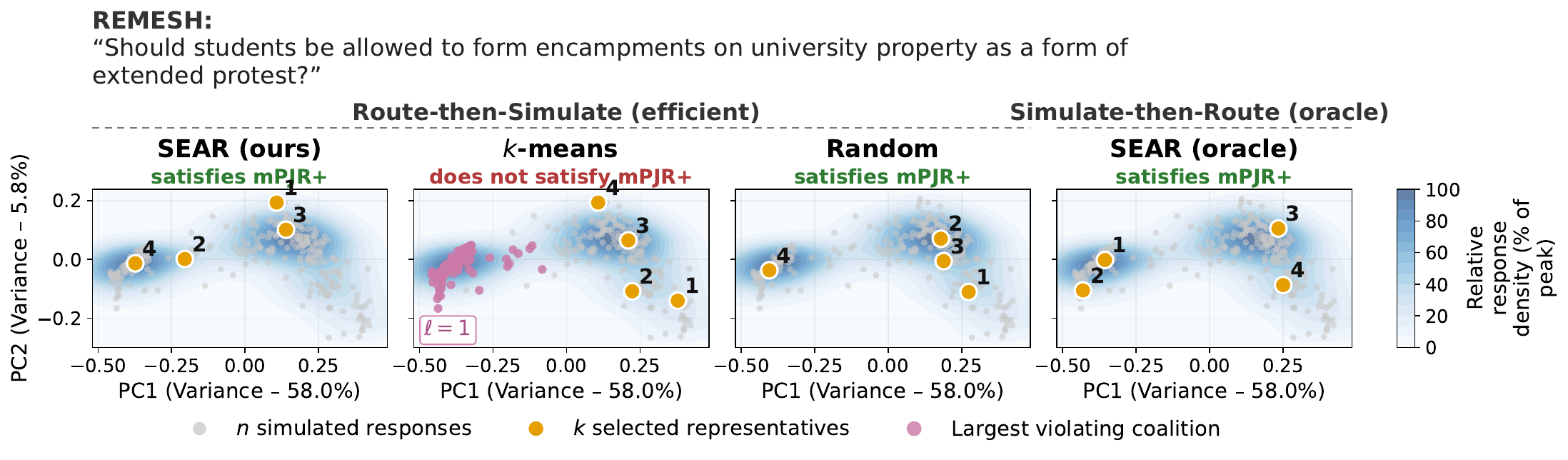}
    \caption{\textbf{Response embeddings and selected representatives for the Remesh prompt ``Should students be allowed to form encampments on university property as a form of extended protest?" ($k=4$).} The text of chosen responses is shown in \Cref{tab:remesh-rank06-pca-slates}.}
    \label{fig:remesh-rank06-pca}
\end{figure*}

{\footnotesize
\begin{longtable}{p{0.22\textwidth} p{0.71\textwidth}}
    \caption{Selected responses for the Remesh prompt ``Should students be allowed to form encampments on university property as a form of extended protest?" ($k=4$). Embeddings are visualized in \Cref{fig:remesh-rank06-pca}.}
    \label{tab:remesh-rank06-pca-slates} \\
    \toprule
    \textbf{Method} & \textbf{Selected responses} \\
    \midrule
    \endfirsthead
    \toprule
    \textbf{Method} & \textbf{Selected responses (continued)} \\
    \midrule
    \endhead
    \begin{minipage}[t]{\linewidth}
    \textbf{SEAR (ours)}\par
    {\scriptsize\textsf{Route-then-Simulate}}\par\smallskip
    {\scriptsize\color{boxgreendark}\bfseries Satisfies mPJR+}
    \end{minipage} &
    \begin{enumerate}[leftmargin=*,nosep,label=\arabic*.]
        \item I don't think students should be allowed to form encampments on university property as a form of extended protest. While I believe in the importance of free speech and peaceful protesting, I think encampments can disrupt the normal functioning of the university and impact the daily lives of other students. It's better to allow students to express their views through peaceful protests and rallies, rather than setting up long-term encampments that can cause disruptions.
        \item I think it's a bit of a gray area, but I would say that as long as the encampments don't become a safety hazard or disrupt the normal functioning of the university, students should be allowed to form them. However, I also think that there should be clear guidelines and regulations in place to ensure that the encampments don't become a nuisance or a disruption to the campus community. Ultimately, it's a balance between allowing students to express themselves and maintaining order on campus.
        \item No, I don't think students should be allowed to form encampments on university property as a form of extended protest. This can lead to disruptions in the normal functioning of the campus and create safety concerns for other students and faculty. If students want to protest, they should be allowed to do so in a peaceful and organized manner, but not in a way that interferes with the daily operations of the university.
        \item I think students should be allowed to form encampments on university property as a form of extended protest, as long as they are peaceful and do not disrupt the normal functioning of the campus. It's a way for students to show their commitment to a cause and to keep the issue in the spotlight, and it can be a powerful tool for raising awareness and sparking change. However, the university should also have a clear plan in place for how to handle the encampment and ensure that it doesn't become a safety hazard or a disruption to the campus community.
    \end{enumerate} \\
    \midrule
    \begin{minipage}[t]{\linewidth}
    \textbf{$k$-means}\par
    {\scriptsize\textsf{Route-then-Simulate}}\par\smallskip
    {\scriptsize\color{boxreddark}\bfseries Does not satisfy mPJR+}
    \end{minipage} &
    \begin{enumerate}[leftmargin=*,nosep,label=\arabic*.]
        \item No, I don't think students should be allowed to form encampments on university property as a form of extended protest. It's a disruption to the normal functioning of the campus and can create a hostile environment for other students. Allowing encampments can also lead to the kind of violence and destruction that we've seen in some of these protests.
        \item No, I don't think students should be allowed to form encampments on university property. It's just another way for them to disrupt the normal functioning of the campus and cause inconvenience to other students. They should be allowed to express their views, but not at the expense of others' education and daily lives.
        \item I don't think students should be allowed to form encampments on university property. It can create a disruption to the normal functioning of the campus and may block public spaces to traffic, which I believe is not acceptable. If they want to protest, they should do it in a peaceful and organized manner, but not in a way that interferes with the daily lives of other students and faculty.
        \item I don't think students should be allowed to form encampments on university property as a form of extended protest. While I believe in the importance of free speech and peaceful protesting, I think encampments can disrupt the normal functioning of the university and impact the daily lives of other students. It's better to allow students to express their views through peaceful protests and rallies, rather than setting up long-term encampments that can cause disruptions.
    \end{enumerate} \\
    \midrule
    \begin{minipage}[t]{\linewidth}
    \textbf{Random}\par
    {\scriptsize\textsf{Route-then-Simulate}}\par\smallskip
    {\scriptsize\color{boxgreendark}\bfseries Satisfies mPJR+}
    \end{minipage} &
    \begin{enumerate}[leftmargin=*,nosep,label=\arabic*.]
        \item No, I don't think students should be allowed to form encampments on university property. It's a disruption to the normal functioning of the university and can create a hostile environment for other students. If they want to protest, they should do it in a peaceful and respectful manner, not by taking over university property.
        \item I strongly disagree with allowing students to form encampments on university property as a form of extended protest. It can lead to disruptions in campus life, create safety concerns, and undermine the learning environment. I believe that peaceful protests should be allowed, but they should be contained within designated areas and not disrupt the normal functioning of the university.
        \item I don't think that's a good idea. Allowing students to form encampments on university property could lead to disruptions and safety concerns, and it's not a productive way for them to make their voices heard. I think they should be allowed to protest in designated areas, but not to the point where it becomes a long-term occupation.
        \item I think students should be allowed to form encampments on university property as a form of extended protest, as long as they are not causing harm to people or property. This can be a powerful way for students to express their concerns and bring attention to their cause, and it can also foster a sense of community and solidarity among protesters. However, I do think that there should be clear guidelines and rules in place to ensure that the encampment is not disrupting the normal functioning of the university.
    \end{enumerate} \\
    \midrule
    \begin{minipage}[t]{\linewidth}
    \textbf{SEAR (oracle)}\par
    {\scriptsize\textsf{Simulate-then-Route}}\par\smallskip
    {\scriptsize\color{boxgreendark}\bfseries Satisfies mPJR+}
    \end{minipage} &
    \begin{enumerate}[leftmargin=*,nosep,label=\arabic*.]
        \item I think students should be allowed to form encampments on university property as a form of extended protest, as long as they are peaceful and not disrupting the normal functioning of the university. It's a way for students to express their opinions and bring attention to their cause, and it's a form of protest that's been used effectively in the past. However, the university should also have clear guidelines and rules in place to ensure that the encampment doesn't become a safety hazard or a disruption to the campus community.
        \item I think students should be allowed to form encampments on university property as a form of extended protest, as long as they're peaceful and not disrupting the normal functioning of the university. It's a way for students to express their views and bring attention to their cause, and as long as they're not causing harm to others, I think it's a legitimate form of protest.
        \item I don't think students should be allowed to form encampments on university property as a form of extended protest. It can disrupt the normal functioning of the university and create a safety concern for other students who may not be involved in the protest. I think it's better to allow peaceful protests to take place in designated areas, but not to allow them to set up long-term encampments.
        \item No, I don't think students should be allowed to form encampments on university property. It's just another form of disruption and it's not a productive way to get their point across. If they want to protest, they should do it in a peaceful and respectful manner, but not by taking over university property and disrupting the lives of other students.
    \end{enumerate} \\
    \midrule
    \bottomrule
\end{longtable}
}

\begin{figure*}[h]
    \centering
    \includegraphics[width=\textwidth]{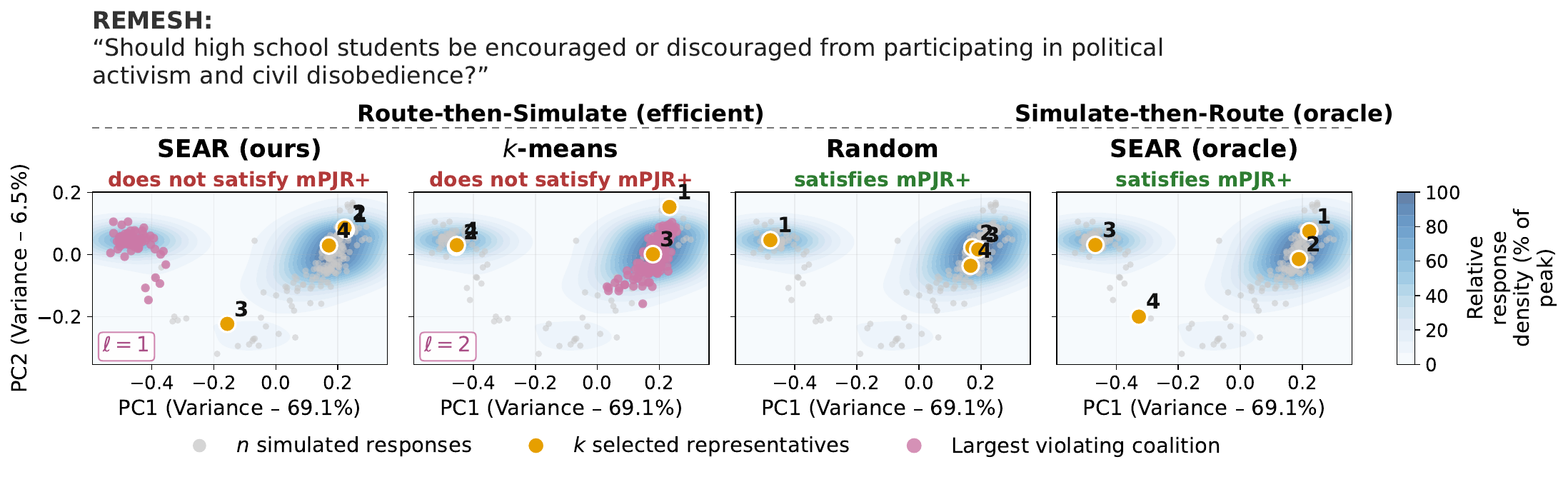}
    \caption{\textbf{Response embeddings and selected representatives for the Remesh prompt ``Should high school students be encouraged or discouraged from participating in political activism and civil disobedience?" ($k=4$).} The text of chosen responses is shown in \Cref{tab:remesh-rank01-pca-slates}.}
    \label{fig:remesh-rank01-pca}
\end{figure*}

{\footnotesize
\begin{longtable}{p{0.22\textwidth} p{0.71\textwidth}}
    \caption{Selected responses for the Remesh prompt "Should high school students be encouraged or discouraged from participating in political activism and civil disobedience?" ($k=4$). Embeddings are visualized in \Cref{fig:remesh-rank01-pca}.}
    \label{tab:remesh-rank01-pca-slates} \\
    \toprule
    \textbf{Method} & \textbf{Selected responses} \\
    \midrule
    \endfirsthead
    \toprule
    \textbf{Method} & \textbf{Selected responses (continued)} \\
    \midrule
    \endhead
    \begin{minipage}[t]{\linewidth}
    \textbf{SEAR (ours)}\par
    {\scriptsize\textsf{Route-then-Simulate}}\par\smallskip
    {\scriptsize\color{boxreddark}\bfseries Does not satisfy mPJR+}
    \end{minipage} &
    \begin{enumerate}[leftmargin=*,nosep,label=\arabic*.]
        \item I think high school students should be encouraged to participate in peaceful political activism and civil disobedience, as it's a great way for them to learn about their rights and develop their critical thinking skills. However, I do believe that they should be taught about the potential consequences of their actions and the importance of respecting the law and the rights of others. By doing so, they can make informed decisions and participate in activism in a responsible and effective manner.
        \item I think high school students should be encouraged to participate in political activism and civil disobedience, as it's a crucial way for them to develop their critical thinking skills, learn about social justice, and express their opinions on important issues. However, I also believe that they should be taught about the potential consequences of their actions and the importance of peaceful and respectful protest. By encouraging activism, we can empower the next generation to become active citizens and drive positive change in their communities.
        \item I think high school students should be encouraged to participate in peaceful and respectful political activism, but discouraged from engaging in civil disobedience, as it can be a slippery slope and may lead to consequences that can impact their future. At this age, they should be learning about the importance of civic engagement and the value of respectful dialogue, rather than resorting to disruptive tactics.
        \item I believe high school students should be encouraged to participate in peaceful political activism and civil disobedience, as it's an essential part of learning about their rights and responsibilities as citizens. However, I also think it's crucial for them to understand the potential consequences of their actions and to engage in respectful dialogue with those who may hold differing views. By doing so, they can develop critical thinking skills, empathy, and a deeper understanding of the complexities of social and political issues.
    \end{enumerate} \\
    \midrule
    \begin{minipage}[t]{\linewidth}
    \textbf{$k$-means}\par
    {\scriptsize\textsf{Route-then-Simulate}}\par\smallskip
    {\scriptsize\color{boxreddark}\bfseries Does not satisfy mPJR+}
    \end{minipage} &
    \begin{enumerate}[leftmargin=*,nosep,label=\arabic*.]
        \item I think high school students should be encouraged to participate in political activism and civil disobedience, as it's a crucial part of learning about their rights and responsibilities as citizens. It's essential for them to develop critical thinking skills and express their opinions on issues that matter to them. By participating in activism, they can gain valuable experience and build confidence in their ability to effect change.
        \item High school students should be discouraged from participating in political activism and civil disobedience. At that age, they're still learning and forming their opinions, and it's not their place to be disrupting the education of others. They should focus on their studies and save their activism for when they're in college or the workforce, where they can make more informed decisions and have a greater impact.
        \item I think high school students should be encouraged to participate in peaceful political activism and civil disobedience, as it's an essential part of learning and growing as individuals. However, I also believe that they should be taught the importance of self-discipline and respect for others during these activities. This will help them navigate complex issues and express their opinions in a constructive manner.
        \item I think high school students should be discouraged from participating in political activism and civil disobedience. At that age, they're still learning and should focus on their education, rather than getting involved in potentially divisive and chaotic protests. It's better for them to develop their critical thinking skills and understand different perspectives before engaging in activism.
    \end{enumerate} \\
    \midrule
    \begin{minipage}[t]{\linewidth}
    \textbf{Random}\par
    {\scriptsize\textsf{Route-then-Simulate}}\par\smallskip
    {\scriptsize\color{boxgreendark}\bfseries Satisfies mPJR+}
    \end{minipage} &
    \begin{enumerate}[leftmargin=*,nosep,label=\arabic*.]
        \item I think high school students should be discouraged from participating in political activism and civil disobedience. At that age, they lack real-world experience and often act impulsively, which can lead to more harm than good. It's better for them to focus on their studies and develop a deeper understanding of the issues before engaging in activism.
        \item I think high school students should be encouraged to participate in peaceful political activism and civil disobedience, as it's an important part of learning about their rights and responsibilities as citizens. However, I also believe that they should be taught the importance of respecting the law and the rights of others, and that their actions should be guided by a sense of empathy and understanding. Ultimately, it's up to parents, teachers, and administrators to guide them in a way that promotes critical thinking and responsible civic engagement.
        \item As a Muslim and a Democrat, I believe that high school students should be encouraged to participate in political activism and civil disobedience, as it is an essential part of civic engagement and learning. By participating in peaceful protests and activism, students can develop critical thinking, leadership, and communication skills, as well as a deeper understanding of the complexities of social and political issues. However, it's essential to ensure that they do so in a safe and responsible manner, with guidance from educators and mentors.
        \item I believe high school students should be encouraged to participate in peaceful political activism and civil disobedience, as it can foster critical thinking, civic engagement, and a sense of social responsibility. However, I also think it's essential for them to understand the potential consequences of their actions and to be guided by responsible and informed decision-making. Ultimately, I think high school students should be encouraged to participate in activism, but with a focus on constructive and respectful dialogue.
    \end{enumerate} \\
    \midrule
    \begin{minipage}[t]{\linewidth}
    \textbf{SEAR (oracle)}\par
    {\scriptsize\textsf{Simulate-then-Route}}\par\smallskip
    {\scriptsize\color{boxgreendark}\bfseries Satisfies mPJR+}
    \end{minipage} &
    \begin{enumerate}[leftmargin=*,nosep,label=\arabic*.]
        \item I think high school students should be encouraged to participate in peaceful political activism and civil disobedience, as it's a great way for them to learn about their rights and develop critical thinking skills. However, I do believe that they should be taught about the potential consequences of their actions and the importance of respecting the law and the rights of others. By doing so, they can make informed decisions and participate in activism in a responsible and effective manner.
        \item I think high school students should be encouraged to participate in peaceful political activism and civil disobedience, as it's a great way for them to learn about their rights and develop their critical thinking skills. However, I also believe that they should be taught about the importance of non-violence and respect for others' opinions, so that they can express themselves effectively and safely. By doing so, they can become more engaged and informed citizens.
        \item I think high school students should be discouraged from participating in political activism and civil disobedience. At their age, they're still learning and should focus on their education rather than getting involved in potentially divisive and confrontational activities. It's better for them to develop their critical thinking skills and understand different perspectives before engaging in activism.
        \item I think high school students should be discouraged from participating in civil disobedience, as it can often lead to disruption and chaos, which can be counterproductive to their cause. However, I do think they should be encouraged to engage in peaceful and respectful forms of political activism, such as writing letters to their representatives, attending town hall meetings, or participating in peaceful protests. This can help them develop important skills and values, such as critical thinking, empathy, and civic engagement.
    \end{enumerate} \\
    \midrule
    \bottomrule
\end{longtable}
}

\subsubsection{Reddit Qualitative Examples}

\begin{figure*}[h]
    \centering
    \includegraphics[width=\textwidth]{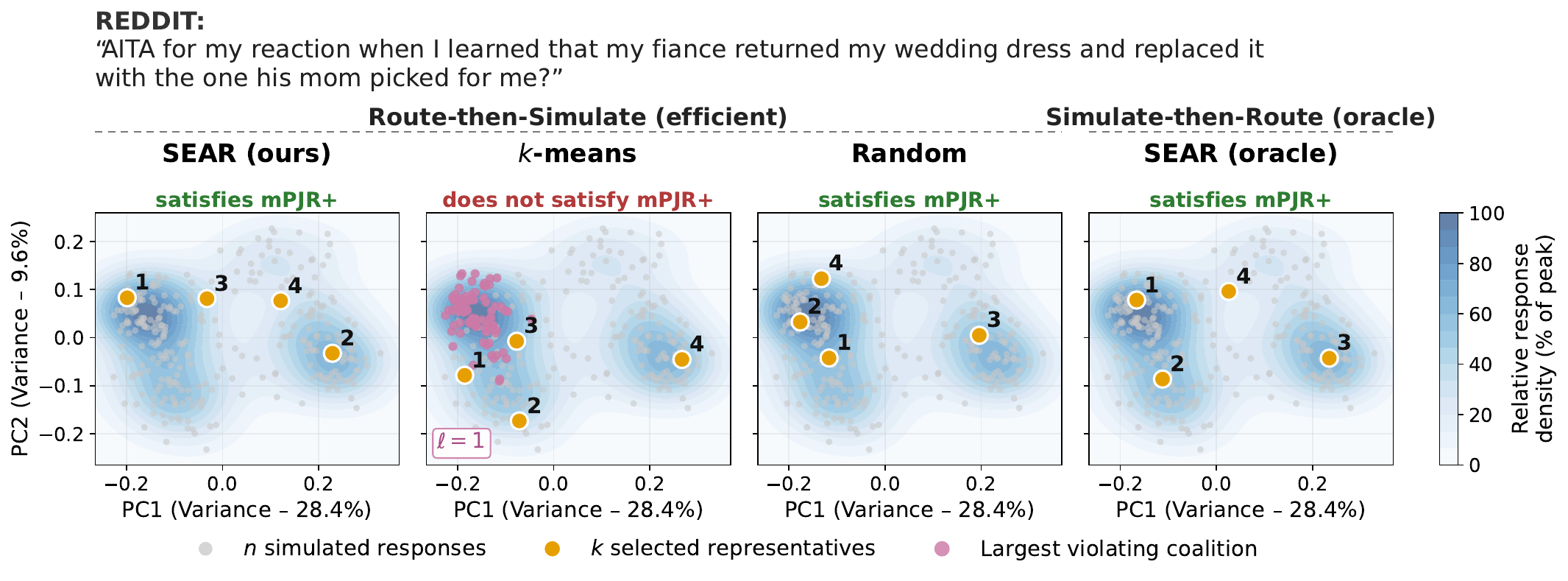}
    \caption{\textbf{Response embeddings and selected representatives for the Reddit prompt ``AITA for my reaction when I learned that my fiance returned my wedding dress and replaced it with the one his mom picked for me?'' ($k=4$).} The text of chosen responses is shown in \Cref{tab:reddit-rank01-pca-slates}.}
    \label{fig:reddit-rank01-pca}
\end{figure*}

{\footnotesize
\begin{longtable}{p{0.22\textwidth} p{0.71\textwidth}}
    \caption{Selected responses for the Reddit prompt ``AITA for my reaction when I learned that my fiance returned my wedding dress and replaced it with the one his mom picked for me?'' ($k=4$). Embeddings are visualized in \Cref{fig:reddit-rank01-pca}.}
    \label{tab:reddit-rank01-pca-slates} \\
    \toprule
    \textbf{Method} & \textbf{Selected responses} \\
    \midrule
    \endfirsthead
    \toprule
    \textbf{Method} & \textbf{Selected responses (continued)} \\
    \midrule
    \endhead
    \begin{minipage}[t]{\linewidth}
    \textbf{SEAR (ours)}\par
    {\scriptsize\textsf{Route-then-Simulate}}\par\smallskip
    {\scriptsize\color{boxgreendark}\bfseries Satisfies mPJR+}
    \end{minipage} &
    \begin{enumerate}[leftmargin=*,nosep,label=\arabic*.]
        \item Your reaction is 100\% justified. You have the right to pick your own wedding dress, and he had no authority to replace it with his mother's choice without your consent. His justification ("it's his wedding, too") ignores your autonomy and personal agency. You deserve respect, not gaslighting about your preferences or being forced to conform to his family's "vision." Screaming at him for undermining your choices is a valid response---no one should dictate your wedding attire or dismiss your preferences. Stay firm: your dress, your choice, your wedding.
        \item NTA for your reaction. You have the right to pick your own wedding dress, not be guilt-tripped into wearing your fianc's mother's "vision." His justification is invalid---he should respect your choices, not side against you. Your anger is justified; he's the one being disrespectful and undermining your autonomy.
        \item NTA. You have every right to choose your own wedding dress---it's your big day, your vision, your choice. He shouldn't prioritize his mom's "vision" over your preferences, especially when the dresses were notably different. Screaming at him for replacing your dress without consent is totally justified; he's the one being unreasonable. Your autonomy and feelings matter, and he should respect that.
        \item Nta. He's the asshole for overriding your choice, gaslighting you about "both dresses being similar," and prioritizing his mom's toxic vision over your wedding day. You have every right to scream---his entitlement and lack of partnership are despicable. Your reaction is justified; he's the one enabling his mom's control and disrespecting your autonomy.
    \end{enumerate} \\
    \midrule
    \begin{minipage}[t]{\linewidth}
    \textbf{$k$-means}\par
    {\scriptsize\textsf{Route-then-Simulate}}\par\smallskip
    {\scriptsize\color{boxreddark}\bfseries Does not satisfy mPJR+}
    \end{minipage} &
    \begin{enumerate}[leftmargin=*,nosep,label=\arabic*.]
        \item I think your reaction is entirely justified. You have the right to choose your own wedding dress, especially since you're paying for it and it's your wedding. His mom's "visions" shouldn't override your preferences or autonomy. He undermined your choice by replacing the dress without discussing it with you, which is disrespectful and dismissive of your input. You have every right to feel upset and react strongly to being disempowered in your own celebration. His insistence on "giving the dress a chance" ignores the fact that you already had a vision for your wedding and that his mom's choice was not aligned with your preferences. This is about respect, boundaries, and mutual agreement in a partnership.
        \item NTA for your reaction! He's the asshole for returning your dress and overriding your choice without consent. You have every right to pick your own wedding dress, not be guilt-tripped into his mom's "vision." His excuse about "both dresses not being that different" is bs---your preferences matter, and he should respect your autonomy. You deserve a wedding that reflects *your* vision, not his mom's control tactics. Stay firm, he should've listened, not replaced your dress and acted like it's his call to make.
        \item Totally AITA for your reaction! He had no right to return your dress without consent, replace it with his's choice, and expect you to "give it a chance" after he undermined your input. You deserve to choose your own wedding dress---your vision, your budget, your style. His mom's "vision" doesn't override your autonomy, especially since you're paying for it. He's the one being toxic by dismissing your preferences and gaslighting you about "his wedding too." You have every right to scream and refuse to comply with his decision.
        \item NTA. He has no right to override your choice of wedding dress, especially since you're paying for it and it's your wedding. Your reaction is justified---no one should force you to conform to their vision, especially not your fianc's mother. You deserve to choose your own dress without guilt or pressure. Stand your ground.
    \end{enumerate} \\
    \midrule
    \begin{minipage}[t]{\linewidth}
    \textbf{Random}\par
    {\scriptsize\textsf{Route-then-Simulate}}\par\smallskip
    {\scriptsize\color{boxgreendark}\bfseries Satisfies mPJR+}
    \end{minipage} &
    \begin{enumerate}[leftmargin=*,nosep,label=\arabic*.]
        \item I think your reaction is 100\% valid! You have every right to choose your own wedding dress, and he undermining your choice by replacing it with his mother's pick is disrespectful. Your anger at being gaslit about your own vision and forced to conform to HIS mom's "vision" is justified. You're not obligated to "give it a chance" if he's not respecting your autonomy. This is about your wedding, your choices, and your right to resist pressure from his family. You deserve support, not guilt-tripping!
        \item Yup, your reaction is 100\% valid. You have every right to be upset---he returned your dress without consulting you, overrode your choice, and insisted his mom's "vision" takes precedence. That's not respecting your autonomy or the fact that you're paying for the wedding dress you want to wear. His mom's input is welcome, but it shouldn't invalidate your preferences or make you feel like you're "turning down" her help. You deserve to choose your own dress, and he should've respected that. Screaming at him for gaslighting you? Totally justified. This is about you having a say in your own wedding, not his mom's fantasy.
        \item NTA for your reaction. You have every right to choose your own wedding dress. He overstepped by unilaterally replacing it without your consent, dismissing your preferences, and gaslighting you about the dress differences. Your anger is justified---his actions disregard your autonomy and boundaries.
        \item I second this motion! Your reaction is 100\% justified. You have every right to choose your own wedding dress---it's *your* wedding, not his mom's vision. He had no right to return your dress and replace it without your consent, especially after you clearly stated your preference and that the dresses were "so much different." His insistence on "his mom's vision" is invalid; you deserve respect for your choices, not guilt-tripping or undermining your input. Screaming at him for gaslighting you? Totally reasonable. He's the one who's overreacting by expecting you to "give the dress a chance" after he made the unilateral decision. You're entitled to your own style and autonomy, and he's the one failing to respect that.
    \end{enumerate} \\
    \midrule
    \begin{minipage}[t]{\linewidth}
    \textbf{SEAR (oracle)}\par
    {\scriptsize\textsf{Simulate-then-Route}}\par\smallskip
    {\scriptsize\color{boxgreendark}\bfseries Satisfies mPJR+}
    \end{minipage} &
    \begin{enumerate}[leftmargin=*,nosep,label=\arabic*.]
        \item **AITA: You are 100\% justified in your reaction.** He has no right to unilaterally replace your chosen wedding dress with his mom'S pick, especially after you clearly communicated your preferences and that the dresses were *significantly* different. Wedding dress choices are personal, funded by you, and should reflect *your* vision, not his mom's. His insistence that "it's his wedding too" ignores your autonomy and the fact that you're the one wearing the dress. Screaming at him for disregarding your feelings? Totally valid. He's the one being entitled here, not you. Stay proud of your stance---your choices matter.
        \item NTA for your reaction. You have every right to choose your own wedding dress and express your preferences. His mom's "vision" doesn't override your autonomy or style. He undermining your choice and replacing your dress without consent is disrespectful. You deserve to wear what makes you feel empowered, not be pressured to conform to someone else's expectations. Your anger is justified---this is about respecting your autonomy and boundaries.
        \item NTA. You have the right to pick your own wedding dress. His mom's "vision" doesn't override your choices. He replaced your dress without consent---red flag. Respect your boundaries, your dress, and your autonomy. You're justified in your reaction.
        \item YTA for him replacing your dress without consent. You have every right to choose your own wedding dress---especially since you're paying for it and it's your wedding. His mom's "visions" don't override your autonomy. Screaming at him for undermining your choice? 100\% justified. He's the one being entitled, not you. Don't "go with flow" if it means sacrificing your own wants. *He's the one who's being the asshole here.*
    \end{enumerate} \\
    \midrule
    \bottomrule
\end{longtable}
}

\begin{figure*}[h]
    \centering
    \includegraphics[width=\textwidth]{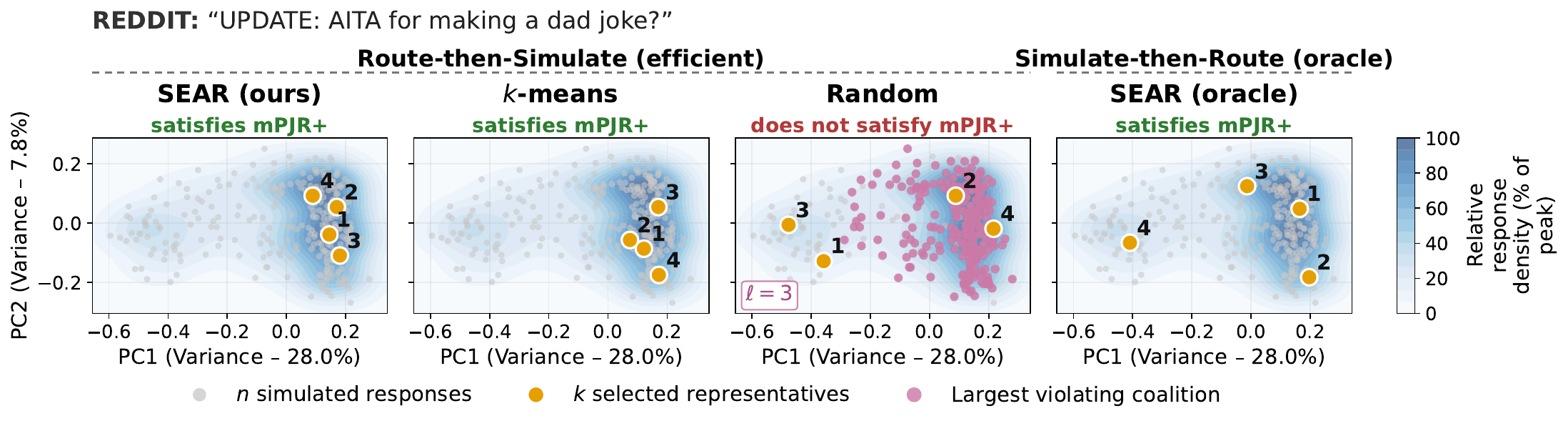}
    \caption{\textbf{Response embeddings and selected representatives for the Reddit prompt ``AITA for making a dad joke?'' ($k=4$).} The text of chosen responses is shown in \Cref{tab:reddit-rank03-pca-slates}.}
    \label{fig:reddit-rank03-pca}
\end{figure*}

{\footnotesize
\begin{longtable}{p{0.22\textwidth} p{0.71\textwidth}}
    \caption{Selected responses for the Reddit prompt ``AITA for making a dad joke?'' ($k=4$). Embeddings are visualized in \Cref{fig:reddit-rank03-pca}.}
    \label{tab:reddit-rank03-pca-slates} \\
    \toprule
    \textbf{Method} & \textbf{Selected responses} \\
    \midrule
    \endfirsthead
    \toprule
    \textbf{Method} & \textbf{Selected responses (continued)} \\
    \midrule
    \endhead
    \begin{minipage}[t]{\linewidth}
    \textbf{SEAR (ours)}\par
    {\scriptsize\textsf{Route-then-Simulate}}\par\smallskip
    {\scriptsize\color{boxgreendark}\bfseries Satisfies mPJR+}
    \end{minipage} &
    \begin{enumerate}[leftmargin=*,nosep,label=\arabic*.]
        \item NTA. The joke was playful and she responded affectionately with a hug. You two had a great bonding time, and she was clearly in on the joke. Keep enjoying your daughter's company!
        \item I'm with you on this! The joke was a playful, bonding moment---she smirked, laughed, and hugged you. If anything, her reaction shows she found it lightheared, not offensive. You two had a great weekend, and her response was a sign of affection, not trauma. No AITA here; it's a sweet, quirky moment between parent and daughter.
        \item NTA. The joke was a playful moment during a bonding weekend where she laughed, hugged you, and clearly didn't take it personally. Your daughter's smirk and subsequent affectionate gesture (hugging, forehead kiss) indicate she found it lightheard, not offensive. No boundary violations or disrespect here---just a family moment with mutual joy.
        \item Yep, the joke was light-hearted, and your bonding time shows you're doing great. The daughter's reaction is hers to process---she's the one who shared the post and apologized, not you. No asshole here, just a playful moment between parent and kid.
    \end{enumerate} \\
    \midrule
    \begin{minipage}[t]{\linewidth}
    \textbf{$k$-means}\par
    {\scriptsize\textsf{Route-then-Simulate}}\par\smallskip
    {\scriptsize\color{boxgreendark}\bfseries Satisfies mPJR+}
    \end{minipage} &
    \begin{enumerate}[leftmargin=*,nosep,label=\arabic*.]
        \item I think the joke was playful and part of bonding, not inappropriate. She planned her response (smirk, hug, kiss) which shows she was engaging, not hurt. Your intent was lightheared, and she's the one who escalated with the "I'm not your" line. No boundary violation here---just a quirky, affectionate exchange.
        \item NTA. The joke was a playful dad joke, she found it funny (smirk, hug, forehead kiss), and the bonding weekend succeeded. She shared the post after apologizing---context matters. No harm, no disrespect. You're not the asshole; she's the one who's clearly into dad jokes and bonding.
        \item I'm with you on this! The joke was a playful, bonding moment---she smirked, laughed, and hugged you. If anything, her reaction shows she found it lightheared, not offensive. You two had a great weekend, and her response was a sign of affection, not trauma. No AITA here; it's a sweet, quirky moment between parent and daughter.
        \item NTA. The joke wasn't toxic, she laughed, hugs, and you both had a great bonding weekend. Playful banter is fine. She's entitled to her witty reply, and you're entitled to a fun moment.
    \end{enumerate} \\
    \midrule
    \begin{minipage}[t]{\linewidth}
    \textbf{Random}\par
    {\scriptsize\textsf{Route-then-Simulate}}\par\smallskip
    {\scriptsize\color{boxreddark}\bfseries Does not satisfy mPJR+}
    \end{minipage} &
    \begin{enumerate}[leftmargin=*,nosep,label=\arabic*.]
        \item Did she feel(mocked/disrespected)? Clarify her reaction to the joke.
        \item Yep, the joke was light-hearted, and your bonding time shows you're doing great. The daughter's reaction is hers to process---she's the one who shared the post and apologized, not you. No asshole here, just a playful moment between parent and kid.
        \item NTA. The joke was shitty, and she had every right to cry. You tried to lighten the mood, but her feelings matter. Accountability: if your daughter's hurt, own it. You apologized once---she deserved more than a smirk and hug. She's entitled to respect, not your dad jokes at her expense.
        \item Same here! The joke was playful and she clearly enjoyed it---smirking, hugging, and all. It sounds like a sweet, bonding moment. She was in on it (planning it for weeks?), so no harm done. Your weekend sounds amazing, and her reaction shows she loved the time with you. No AITA here; it was a lighthearted, fun exchange!
    \end{enumerate} \\
    \midrule
    \begin{minipage}[t]{\linewidth}
    \textbf{SEAR (oracle)}\par
    {\scriptsize\textsf{Simulate-then-Route}}\par\smallskip
    {\scriptsize\color{boxgreendark}\bfseries Satisfies mPJR+}
    \end{minipage} &
    \begin{enumerate}[leftmargin=*,nosep,label=\arabic*.]
        \item The joke was a dad joke, not an asshole move. She smirked, hugged you, and played along---clearly she found it fun. The bonding weekend was awesome, and her reaction shows she's in on the joke. No harm, no foul.
        \item NTA. The joke wasn't disrespectful---it was a playful moment during a bonding weekend. She smirked, hugged you, and you kissed her forehead. Mutual joy, not a boundary violation. She's the one who turned it into a fun exchange, not you.
        \item I mean, the joke was playful, the daughter smirked, hugs, and you kissed her forehead---sounds like she enjoyed the banter. If she cried again after your apology, maybe she was emotional about the bonding time, not the joke. You're not the asshole; it was a lightheARTED moment.
        \item NTA for the joke, but the intent vs impact matters. The daughter's emotional reaction shows the joke hurt her, even if unintentional. Apologizing is good, but you need to acknowledge her feelings and avoid similar jokes in the future. Bonding is great, but respect her boundaries and emotions.
    \end{enumerate} \\
    \midrule
    \bottomrule
\end{longtable}
}

\begin{figure*}[h]
    \centering
    \includegraphics[width=\textwidth]{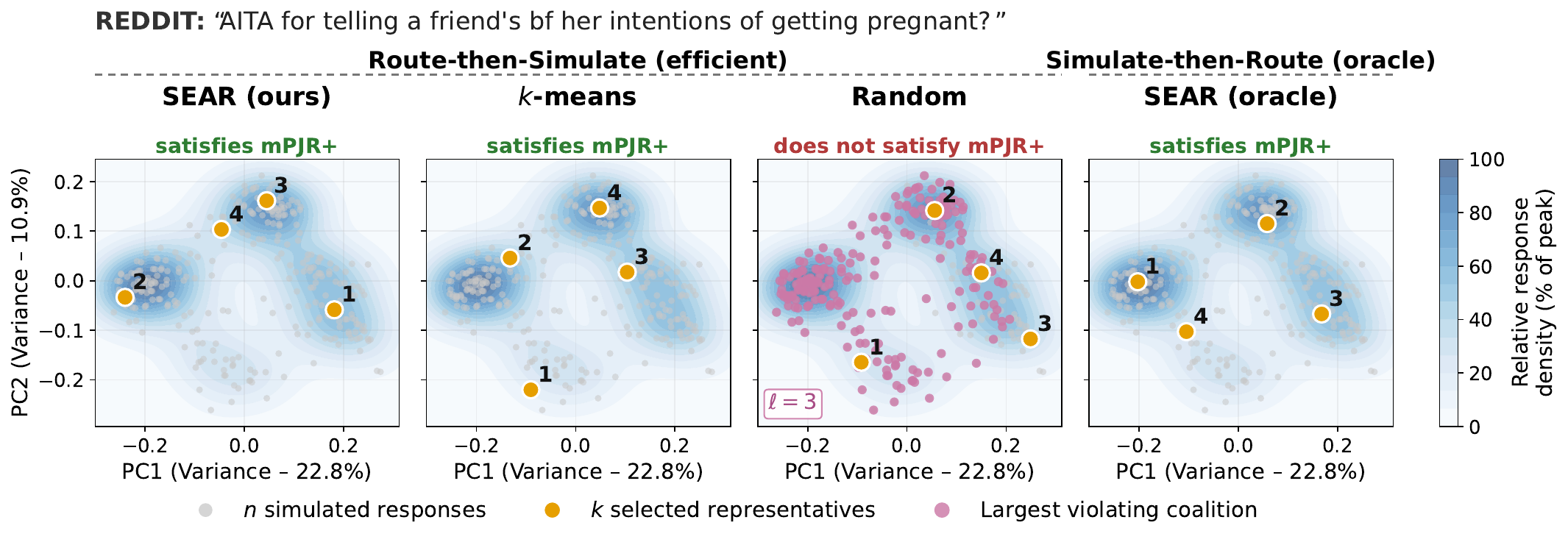}
    \caption{\textbf{Response embeddings and selected representatives for the Reddit prompt ``AITA for telling a friend's boyfriend her intentions of getting pregnant?'' ($k=4$).} The text of chosen responses is shown in \Cref{tab:reddit-rank06-pca-slates}.}
    \label{fig:reddit-rank06-pca}
\end{figure*}

{\footnotesize
\begin{longtable}{p{0.22\textwidth} p{0.71\textwidth}}
    \caption{Selected responses for the Reddit prompt ``AITA for telling a friend's boyfriend her intentions of getting pregnant?'' ($k=4$). Embeddings are visualized in \Cref{fig:reddit-rank06-pca}.}
    \label{tab:reddit-rank06-pca-slates} \\
    \toprule
    \textbf{Method} & \textbf{Selected responses} \\
    \midrule
    \endfirsthead
    \toprule
    \textbf{Method} & \textbf{Selected responses (continued)} \\
    \midrule
    \endhead
    \begin{minipage}[t]{\linewidth}
    \textbf{SEAR (ours)}\par
    {\scriptsize\textsf{Route-then-Simulate}}\par\smallskip
    {\scriptsize\color{boxgreendark}\bfseries Satisfies mPJR+}
    \end{minipage} &
    \begin{enumerate}[leftmargin=*,nosep,label=\arabic*.]
        \item NTA. Ashley's manipulative, dishonest tactics toward Chris---fertility spells, condom sabotage, and lying about her intentions---deserve accountability. You have every right to supportChris's autonomy and integrity. She's the one playing games, not you. Dump her drama and stand with Chris.
        \item You are *not* the asshole. Ashley is the one who acted manipulatively, dishonestly, and irresponsibly. Chris had the absolute right to know that she was plotting to force a pregnancy on him against his explicit wishes. He's entitled to respect, honesty, and boundaries in his relationship. Ashley's behavior---joking about "one way or another" getting a baby from Chris, casting fertility spells, and sabotaging condoms---was a clear violation of trust and respect for his autonomy. You had every right to expose her deceit, especially after she admitted her intentions. She should own her own choices, not expect others to cover for her manipulative tactics or impose her desires onto Chris. Chris deserves a partner who respects his terms, not someone who tries to undermine them. You protected him and called out the truth---full stop. She's the one who's being unfair, not you.
        \item You're not the ass-hole. Ashley is the manipulative one trying to force Chris into a situation he explicitly said he wouldn't accept. She lied about her intentions, sabotaged their relationship, and you had a right to expose her deceit. Chris deserves transparency, not being used as a sperm bank.
        \item You're not the asshole; Ashley is the one manipulating Chris and trying to force a baby on him against his wishes. She had no right to plot against his autonomy or ruin their relationship. You had every right to inform him about her intentions and the evidence. She's the one who ruined something fun for Chris and her kids---she should face the consequences.
    \end{enumerate} \\
    \midrule
    \begin{minipage}[t]{\linewidth}
    \textbf{$k$-means}\par
    {\scriptsize\textsf{Route-then-Simulate}}\par\smallskip
    {\scriptsize\color{boxgreendark}\bfseries Satisfies mPJR+}
    \end{minipage} &
    \begin{enumerate}[leftmargin=*,nosep,label=\arabic*.]
        \item NTA, OP. You have every right to informChris about Ashley's manipulative plans to get pregnant against his wishes, especially when she was lying about her intentions and using unethical methods (spells, tampering condoms). She deceived him, violated trust, and acted like a "friend" while plotting to force parenthood on him. You didn't "put her out on the street"---she chose to lie, break promises, and expect Chris to accept her kids *and* a baby she wasn't honest about. Her anger is misplaced; she's the one who acted selfishly and dishonestly. You protected Chris's autonomy and transparency---*that's* the ethical choice.
        \item You're not the asshole; *Ashley* is the one playing manipulative games with Chris's body and autonomy. She lied about her intentions, sought magical shortcuts to force a baby on him, and tried to tamper with his condoms. You had every right to expose her dishonesty and protect Chris from being used as a vessel for her agenda. Her consequences? She's the one who's irresponsible and entitled, not you.
        \item NTA. Ashley's manipulative tactics and disrespect for Chris's boundaries make her the asshole. You had every right to expose deceit---Chris deserves honesty, not a partner who plots to force a baby he explicitly said was a dealbreaker. Her "jokes" were a lie, and you're justified in protecting his autonomy. She's the one who chose to act disrespectfully.
        \item You're not the asshole. Ashley was being manipulative, dishonest, and gaslighting Chris. She lied about her intentions, tried to sabotage his fertility, and expected him to accept her plans without consent. You had a right to inform him---she's the one who's being toxic, not you. Plus, he deserves transparency, not her tricks.
    \end{enumerate} \\
    \midrule
    \begin{minipage}[t]{\linewidth}
    \textbf{Random}\par
    {\scriptsize\textsf{Route-then-Simulate}}\par\smallskip
    {\scriptsize\color{boxreddark}\bfseries Does not satisfy mPJR+}
    \end{minipage} &
    \begin{enumerate}[leftmargin=*,nosep,label=\arabic*.]
        \item NTA. You have every right to call her out. Ashley's manipulative behavior---joking about forcing a baby on Chris, who explicitly said it's a dealbreaker---disrespects his autonomy and the mutual agreement they had. You're not the asshole; she's the one playing games with his boundaries. Chris deserves transparency, not her sneaky fertility spells and condom-poking. You're supporting a healthy partnership by holding her accountable. Go you for standing up for him!
        \item Wow, you're not the asshole. Ashley's manipulative, dishonest antics (fertility spells, sabotaging condoms) were toxic. Chris has every right to know she'd force a baby on him. You exposed her deceit---valid, not bonkers. She's the one gaslighting, not you.
        \item NTA! You have a right to expose manipulative intentions. Ashley's plans were shady, Chris deserves honesty, and you shouldn't stay silent when someone's plotting to force a baby on a partner who explicitly said no. You've got this!
        \item NTA. Ashley manipulated Chris, lied about her intentionsand tried to force a baby on him. She's the asshole for deceit and exploitation. Consequences are fair.
    \end{enumerate} \\
    \midrule
    \begin{minipage}[t]{\linewidth}
    \textbf{SEAR (oracle)}\par
    {\scriptsize\textsf{Simulate-then-Route}}\par\smallskip
    {\scriptsize\color{boxgreendark}\bfseries Satisfies mPJR+}
    \end{minipage} &
    \begin{enumerate}[leftmargin=*,nosep,label=\arabic*.]
        \item You are NOT the asshole. Ashley's manipulations, dishonesty, and attempts to force a pregnancy on Chris (who explicitly stated he doesn't want kids) are *gross*, *manipulative*, and *disrespectful to Chris's autonomy*. You had a right to inform him about her intentions, especially when she was lying about her "jokes" and actively trying to sabotage their relationship. Accountability matters---she needs to face consequences for her deceit and exploitation. Chris deserves honesty, not a partner who plots to violate his boundaries. You acted with integrity, not malice.
        \item You're not the asshole. Ashley's manipulative tactics (fertility spells, sabotaging condoms) violate Chris's autonomy and consent. You had every right to expose her deceit---she's the one playing games with his body and future. He deserves transparency, not a partner who plots to force a baby on him.
        \item NTA. Ashley was manipulatively plotting to force a baby on Chris, who explicitly stated he doesn't want kids. You had a right to inform him about her deceit. She's the one being dishonest, not you. Chris deserves transparency and respect for his boundaries.
        \item **NTA**. Ashley is the asshole for deceitfully manipulating Chris, who explicitly stated he doesn't want kids. Her fertility spells, condom sabotage, and lying about her intentions to force a pregnancy are *manipulative, dishonest, and exploitative*. You had every right to expose her lies---Chris deserves transparency, not a partner who treats him as an ATM or plots to violate his boundaries. Protecting him from her deceit isn't interference; it's accountability.
    \end{enumerate} \\
    \midrule
    \bottomrule
\end{longtable}
}

\end{document}